\RequirePackage{tikz}
\documentclass[default,pdflatex]{sn-jnl}
\usepackage[utf8]{inputenc}
\usepackage{hyperref}

\paperwidth=\dimexpr \paperwidth + 2cm\relax
\oddsidemargin=\dimexpr\oddsidemargin + 1cm\relax
\evensidemargin=\dimexpr\evensidemargin + 1cm\relax
\marginparwidth=\dimexpr \marginparwidth + 1cm\relax

\usepackage{times}  %
\usepackage{helvet}  %
\usepackage{courier}  %
\usepackage{graphicx} %
\usepackage{natbib}  %
\usepackage{caption} %

\usepackage{multirow}
\usepackage{tabularx}
\usepackage{amsmath}
\usepackage{amsthm}
\usepackage{amssymb,amsfonts,textcomp}
\usepackage{latexsym}
\usepackage{multirow}
\usepackage{microtype}
\usepackage{rotating}
\usepackage{url}
\usepackage{comment}
\usepackage{paralist}
\usepackage{booktabs}
\usepackage{bm}
\usepackage[]{units}
\usepackage{mathtools}
\usepackage{cleveref}

\usepackage{thmtools}
\usepackage{thm-restate}
\usepackage{amsthm}
\usepackage[textsize=tiny]{todonotes}
\newtheorem{theorem}{Theorem}

\newtheorem{proposition}[theorem]{Proposition}
\newtheorem{observation}[theorem]{Observation}
\newtheorem{corollary}[theorem]{Corollary}
\newtheorem{example}{Example}
\newtheorem{definition}{Definition}
\usetikzlibrary{decorations.pathreplacing,calligraphy,patterns.meta}

\newcommand{\NP}{NP}

\newcommand{\ceil}[1]{ \lceil #1 \rceil }
\newcommand{\floor}[1]{ \left\lfloor #1 \right\rfloor }

\usepackage{tikz}
\usepackage{subcaption}
\usetikzlibrary{positioning,arrows,fit}\tikzset{
  block/.style={
    draw,
    rectangle,
    minimum height=0.5cm,
    minimum width=3.5cm,
    align=center
  },
  subblock/.style={
    draw,
    rectangle,
    minimum height=.75cm,
    minimum width=1.5cm,
    align=center
  },
  line/.style={->,>=latex'}
}

\normalbaroutside %

\newcommand{\clrstr}{20}
\newcommand{\N}{N} %
\newcommand{\cC}{C} %
\newcommand{\s}{W} %
\newcommand{\xl}{a_{i}} %
\newcommand{\xr}{b_{i}} %

\newcommand{\dist}{\text{dist}} %
\newcommand{\quickset}[1]{\{ #1 \}} %

\DeclareMathOperator*{\argmax}{arg\,max}

\usepackage{soul}
\usepackage{etoolbox}
\makeatletter
\patchcmd{\SOUL@ulunderline}{\dimen@}{\SOUL@dimen}{}{}
\patchcmd{\SOUL@ulunderline}{\dimen@}{\SOUL@dimen}{}{}
\patchcmd{\SOUL@ulunderline}{\dimen@}{\SOUL@dimen}{}{}
\newdimen\SOUL@dimen
\makeatother
\sethlcolor{blue!30}
\soulregister\cite7
\soulregister\ref7
\soulregister\Cref7
\soulregister\citet7
\soulregister\citep7
\soulregister\footnote8
\newcommand{\new}[1]{#1}

\title{Individual Representation in Approval-Based Committee Voting}

\author[1,2]{\fnm{Markus} \sur{Brill}}\email{markus.brill@warwick.ac.uk}
\author[1]{\fnm{Jonas} \sur{Israel}}\email{j.israel@tu-berlin.de}
\author[3,4]{\fnm{Evi} \sur{Micha}}\email{pmicha@usc.edu}
\author[1,5]{\fnm{Jannik} \sur{Peters}}\email{jannik.peters@tu-berlin.de}
\affil[1]{\orgdiv{Research Group Efficient Algorithms}, \orgname{TU Berlin}, \orgaddress{\country{Germany}}}
\affil[2]{\orgdiv{Department of Computer Science}, \orgname{University of Warwick}, \orgaddress{\country{UK}}}
\affil[3]{\orgdiv{Department of Computer Science}, \orgname{Harvard University}, \orgaddress{\country{USA}}}
\affil[4]{\orgdiv{Department of Computer Science}, \orgname{University of Southern California}, \orgaddress{\country{USA}}}
\affil[5]{\orgdiv{School of Computing}, \orgname{National University of Singapore}, \orgaddress{\country{Singapore}}}

\begin{document}

\abstract{
   When selecting multiple candidates based on approval preferences of voters, the proportional representation of voters' opinions is an important  and  well-studied desideratum. Existing criteria for evaluating the representativeness of outcomes focus on \textit{groups} of voters and demand that sufficiently large and cohesive groups are ``represented'' in the sense that candidates approved by \textit{some} group members are selected.  Crucially, these criteria say nothing about the representation of \textit{individual} voters, even if these voters are members of groups that deserve representation. 
   In this paper, we formalize the concept of \textit{individual representation (IR)} and explore to which extent, and under which circumstances, it can be achieved. 
We show that checking whether an IR outcome exists is computationally intractable, and we verify that all common approval-based voting rules may fail to provide IR even in cases where this is possible. We then focus on domain restrictions and establish an interesting contrast between ``voter interval'' and ``candidate interval'' preferences. This contrast can also be observed in our experimental results, where we analyze the attainability of IR for realistic preference profiles.
}

\maketitle

\section{Introduction}

We consider the problem of selecting a fixed-size subset of candidates (a so-called \textit{committee}) based on the approval preferences of voters. This problem has been extensively studied in recent years \citep[][]{LaSk22a} and has a wide variety of applications, including political elections \citep{BGP+24a}, recommender systems \citep{SLB+17a},
medical diagnostic decision-making \citep{medical}, 
\new{blockchain consensus protocols~\citep{BBC+24a},}
and participatory budgeting~\citep{PPS21a}. 

A central concern in committee voting is the principle of  \textit{proportional representation}, which states that the voters' interests and opinions should be reflected proportionately in the committee. 
While proportional representation is intuitive to understand in scenarios such as apportioning parliamentary seats based on vote shares \citep{BaYo82a,Puke14a}, it is less straightforward to formalize in the context of approval-based committee elections. Indeed, the literature has defined a number of different concepts aiming to capture proportional representation \citep{ABC+16a,SFF+17aNEW,PeSk20b,Skow21a,PPS21a,BrPe23a}. 

Most (if not all) of these approval-based proportionality notions focus on the representation of \textit{groups} of voters. Specifically, it is usually required that each sufficiently large group of voters is ``represented'' in the committee,\footnote{Often, there is also a condition on the ``cohesiveness'' of the group, stating that the approval preferences of group members need to be sufficiently aligned. \new{This requirement is extensively discussed by \citet{BrPe23a}.}}
 where the interpretation of ``representation'' differs across different notions. For example, \textit{extended justified representation} \cite{ABC+16a} prescribes that there exists at least one voter in the group approving a certain number of committee members, whereas \textit{proportional justified representation} \citep{SFF+17aNEW} demands that there are sufficiently many committee members that are each approved by at least one voter in the group. %
Notably, neither definition comprises any representation requirements for  \textit{individual} voters in a group.
Thus, a group may count as ``represented'' even though some voters in the group do not approve a single committee member.\footnote{\new{Axioms like extended justified representation offer significant lower bounds on the \textit{average} satisfaction of a voter group (e.g., a high proportionality degree \citep{Skow21a}). However, this still does not ensure representation of voters on the individual level.}}

In this paper, we adopt an {\em individualistic} point of view:  our goal is to provide all members of a voter group equal guarantees. 
Intuitively, when a population consists of $n$ voters and a committee of $k$ representatives is elected, we expect every cohesive voter group of size $\ell \cdot n/k$ to be represented by $\ell$ representatives in the committee; thus, each individual group member might reasonably hope that at least $\ell$ candidates represent her in the committee. This notion, which we call \textit{individual representation}, is aligned with the notion of ``individual fairness'' that was recently introduced in clustering (and in particular in facility location problems) by \citet{clust-fair-2020}: there, each individual expects to be served by a facility in distance proportional to the radius of the ball that captures its $n/k$ closest neighbors, where~$n$ is the number of individuals and $k$ is the number of facilities.

Individual representation, as defined in this paper, is a strengthening of a notion called \textit{semi-strong justified representation} by \citet{ABC+16a}. The latter property requires that all members of a group are represented in the committee \textit{at least once}, given that the group is large and cohesive enough. Individual representation strengthens this requirement by demanding that all members of cohesive groups are represented multiple times (in proportion to the group size).  
\citet{ABC+16a} observed that semi-strong JR cannot be provided in all instances; this immediately implies our stronger requirement is not universally attainable either.

In this paper, we systematically study \textit{individual representation (IR)}. Notwithstanding the observation that IR demands cannot always be met, we clarify how IR relates to existing axioms and we show that a large range of common approval-based committee voting rules can fail to provide IR even in cases where IR is achievable.
We observe that even committees \textit{approximating} IR may fail to exist. Moreover, we answer a question by \citet{ABC+16a} by showing that it is computationally intractable to decide whether a given instance admits a committee providing semi-strong JR or individual representation.
We then turn our attention to \textit{restricted domains} of preferences \citep{ElLa15a,yang2019tree} and demonstrate that positive results can be obtained. Doing so, we uncover a striking difference between the \textit{candidate interval} and \textit{voter interval} domains: whereas the former restriction does not admit any non-trivial approximation of IR, we devise an efficient algorithm for selecting committees approximating IR for the latter.
This is surprising insofar as these two domain restrictions often exhibit similar behavior \citep{PiSk22a,terzopoulou2021restricted}.\footnote{\new{A notable exception is the work by \citet{Pete18b}, who derives polynomial-time algorithms for the candidate interval domain, but not for the voter interval domain.}}
Finally, we experimentally study how often IR is achievable for a wide variety of generated preference data, and how often established  voting rules select IR outcomes.

\section{Preliminaries}
\label{sec:prelim}
For $t \in \mathbb{N}$, we let $[t]$ denote the set $\{1, 2,\dots, t\}$.
Let  $\N=[n]$ be a set of $n$ \emph{voters} %
and $\cC=\{c_1,\dots,c_m\}$ be a set of $m$ \emph{candidates}. 
Each voter $i \in \N$ approves a subset $A_i \subseteq \cC$ of candidates. An \emph{(approval) profile} %
$A = (A_1, \dots, A_n)$ contains the approval set $A_i$ of each voter $i \in N$. 
\new{We often illustrate approval profiles graphically, see Figures \ref{fig:noIR}, \ref{fig:IF-consistency_counterexamples}, and~\ref{fig:SSJR-counterexamples}.}

Given a committee size $k \in [m]$,
we want to select a subset $\s \subseteq \cC$ of size $|\s|=k$, referred to as a \textit{committee}. 
We call $(A,k)$ an \textit{approval-based committee (ABC) election}. An \textit{ABC voting rule} takes as input an ABC election $(A,k)$ and outputs one or more committees of size~$k$.

As is standard in the ABC election literature, we assume that voters only care about the number of approved candidates in the committee, i.e., voter $i$ evaluates a committee $\s$ by $\lvert \s \cap A_i \rvert$. 
Given a subset $S \subseteq \cC$ of candidates, we let $N(S)$ denote the set of voters who approve all candidates in $S$, i.e., $N(S) =  \{i \in \N \colon S \subseteq A_i\}$.

Given an ABC election $(A,k)$ and $\ell \in \mathbb{N}$, we call a group $V \subseteq \N$ of voters  \mbox{\emph{$\ell$-cohesive}} if 
$|V|\geq \ell \cdot  \frac{n}{k}$ and $\lvert \bigcap_{i \in V} A_i  \rvert \geq \ell$.
The following representation notions are due to \citet{ABC+16a} and \citet{SFF+17aNEW}.

\begin{definition}
Consider an ABC election $(A,k)$.
A committee $\s \subseteq C$ of size $k$ provides 
\begin{itemize}
    \item \emph{justified representation (JR)} if for each $1$-cohesive group $V \subseteq \N$, there is a voter $i \in V$ with $|W \cap A_i| \geq 1$;
    \item \emph{proportional justified representation (PJR)} if for each $\ell \in \mathbb{N}$ and each $\ell$-cohesive group $V \subseteq N$, it holds that $|W \cap (\bigcup_{i \in V} A_i)| \geq \ell$;
    \item \emph{extended justified representation (EJR)} if for each $\ell \in \mathbb{N}$ and each \mbox{$\ell$-cohesive} group $V \subseteq \N$, there is a voter $i \in V$ with $|W \cap A_i| \geq \ell$;
    \item \emph{core stability} if for each  group $V \subseteq N$ (independent of $V$ being $\ell$-cohesive) and $S \subseteq C$ with $|V| \ge |S| \cdot \frac{n}{k}$, there is a voter $i \in V$ with $|W \cap A_i| \geq |S \cap A_i|$.
\end{itemize}
\end{definition}

It is well-known that core stability implies EJR, which in turn implies PJR, which implies~JR \citep{ABC+16a,SFF+17aNEW}.
All of these notions have in common that they consider a group of voters ``represented'' as long as at least one voter in the group is sufficiently represented. %
This point of view might be hard to justify in many contexts. In the following section, we present our approach to representation that takes into account every voter in a group individually.

\section{Individual Representation}
\label{sec:ind_rep}

In this section, we define the main concept of this paper: individual representation. This notion builds on the idea of (semi-)strong justified representation as defined by \citet{ABC+16a} and the notion of individual fairness in clustering as defined by  \citet{clust-fair-2020}.

Similarly to the proportionality notions defined in \Cref{sec:prelim}, we assume that a voter deserves some representation in an ABC election if she can find enough other voters who all approve a subset of candidates in common. This follows the rationale that every member of a group of voters that \textit{(i)} makes up a sizable part of the electorate and 
\textit{(ii)} can come to an agreement on how (part of) the committee ought to be filled, should be represented accordingly.

Given an ABC election $(A,k)$, we determine the number of seats that voter $i \in N$ can justifiably demand as
\begin{align*}
    f_i \coloneqq \max_{S \subseteq A_i} \{|S| \colon |N(S)| \geq {|S|\cdot n/k }\}.
\end{align*}
In words, $f_i$ is the largest value $f$ such that voter $i$ can find enough like-minded voters to form an $f$-cohesive group.
\new{In particular, $f_i = 0$ for all voters $i$ who are not contained in any cohesive group of size at least $n/k$.}

\begin{definition}[Individual Representation]
Given an ABC election $(A,k)$, a committee $\s \subseteq C$ of size $|W|\le k$ provides \emph{individual representation (IR)} if $|\s \cap A_i| \geq f_i$ \; for all voters $i \in \N$.
\end{definition}  

When only requiring $\s \cap A_i \neq \emptyset$ for every voter with $f_i > 0$, we get \emph{semi-strong justified representation (semi-strong JR)} as defined by \citet{ABC+16a}. The authors of that paper provide an example showing that semi-strong JR committees do not always exist \new{(see Figure \ref{fig:noIR})}. Since individual representation clearly is a more demanding property, it immediately follows that \textit{IR committees} (i.e., committees providing IR) do not need to exist either. 

\begin{observation}
There exist instances of ABC elections that do not admit an IR committee.
\end{observation}

\begin{figure}
    \centering
        \scalebox{.75}{
    \begin{tikzpicture}
    [yscale=0.4,xscale=0.65,
    voter/.style={anchor=south}]
    
        \foreach \i in {1,...,9}
    		\node[voter] at (\i-0.5, -1.5) {$\i$};
        
        \draw[fill=orange!\clrstr] (0, 0) rectangle (3, 1);
        
        \draw[fill=violet!\clrstr] (2, 1) rectangle (5, 2);
        
        \draw[fill=teal!\clrstr] (4, 0) rectangle (7, 1);

        \draw[fill=red!40] (6, 1) rectangle (9, 2);
        
        \node at ( 1.5, 0.5) {$c_{1}$};
        \node at ( 5.5, 0.5) {$c_{3}$};
        \node at ( 3.5, 1.5) {$c_{2}$};
        \node at ( 7.5, 1.5) {$c_{4}$};
    \end{tikzpicture}}
    \caption{\new{Approval profile showing that IR committees do not
always exist. Voters correspond to integers and approve all candidates placed above them. %
For $k=3$, we have $n/k=3$ and $f_i=1$ for each voter $i \in [9]$. Clearly, there is no $W\subseteq \{c_1,c_2,c_3,c_4\}$ of size $|W|\le 3$ that satisfies $|W\cap A_i|\ge 1$ for all $i$. This instance appears in the paper by \citet{ABC+16a} as Example 7.}}
    \label{fig:noIR}
\end{figure}

One immediate follow-up question is whether we can guarantee IR in an approximate sense. To study this question, we introduce the notion of $(\alpha, \beta)$-individual representation, which  uses additive \textit{and} multiplicative approximation parameters. 
\begin{definition}[$(\alpha,\beta)$-IR]
Given an ABC election $(A,k)$, a committee $\s \subseteq C$ of size at most $k$ provides \emph{$(\alpha, \beta)$-individual representation ($(\alpha,\beta)$-IR)} if for every voter $i \in \N$ it holds that
   $\alpha \cdot |A_i \cap \s| + \beta \geq  f_i $, with $\alpha\geq 1$ and $\beta \geq 0$. 
\end{definition}
\noindent

Unfortunately, non-trivial approximation guarantees are impossible to obtain without restricting the set of profiles.
\begin{theorem}\label{thm:no-beta-approx} 
For every $k\ge 2$, there exists an instance $(A,k)$ that does not admit an $(\alpha, \beta)$-IR committee for $\beta < k-1$, and any $\alpha\geq 1$.
\end{theorem}

\vspace{-0.5cm}

\begin{proof}
Fix $k\ge 2$ and let $n = k\cdot (k+1)$. Note that $n/k = k+1 > k$. Consider the profile in which for each voter $i \in [n/k]$, we have  $A_i=\{c_{(k-1)\cdot (i-1)+1},\dots, c_{(k-1)\cdot i}\}$ and all remaining $n-n/k$ voters approve all candidates. That is, the approval sets of the first $n/k$ voters are pairwise disjoint and contain $k-1$ candidates each.

For every voter $i \in [n/k]$ we get that $\lvert N(A_i) \rvert = 1 + (n-n/k) = 1 + (k-1)(k+1)$. Since $n/k = k+1$, this implies that $f_i = k-1$ for all $i \in [n/k]$. Further, for all distinct voters $i, i' \in [n/k]$ it holds that $A_i \cap A_{i'} = \emptyset$. However, since $n/k>k$, for each $\s \subseteq \cC$ with $|\s|\leq k$ there is a voter $i \in [n/k]$ with $|A_i \cap \s|=0$. Thus, for any $\alpha \geq 1$ and $\beta < k-1$, this instance does not admit an $(\alpha, \beta)$-IR committee.
\end{proof}

To see that \new{this bound on $\beta$} is the worst-possible, note that if $f_i=k$ for some voter $i$, this means that \textit{all} voters have a set of at least $k$ jointly approved candidates (and a committee consisting of such candidates would provide IR). On the other hand, every committee trivially provides $(1,k-1)$-IR whenever $f_i<k$ for all $i\in N$.
We study approximation bounds for $(\alpha,\beta)$-IR on restricted domains in \Cref{sec:dom_rest}.

\subsection{Relation to other Proportionality Axioms}

We have already observed that IR is a strengthening of semi-strong JR. Furthermore, it is easy to see that every IR committee also provides EJR (and thus PJR and JR). On the other hand, there exist profiles where semi-strong JR committees exist that do not provide PJR.
To build intuition on how IR differs from the other notions, and on how it leads to the election of committees that might be considered ``fair'' from an individual voter perspective, consider the following two examples, illustrated in \Cref{fig:IF-consistency_counterexamples}.

\definecolor{mypink1}{rgb}{0.858, 0.188, 0.478}

\begin{figure}
    \centering
        \scalebox{.75}{
    \begin{tikzpicture}
    [yscale=0.4,xscale=0.65,
    voter/.style={anchor=south}]
    
        \foreach \i in {1,...,8}
    		\node[voter] at (\i-0.5, -1.5) {$\i$};
        
        \draw[fill=orange!\clrstr] (0, 0) rectangle (4, 1);
        
        \draw[fill=violet!\clrstr] (4, 0) rectangle (8, 1);
        
        \draw[fill=teal!\clrstr] (1, 1) rectangle (7, 2);
        
        \node at ( 2, 0.5) {$c_{1}$};
        \node at ( 6, 0.5) {$c_{2}$};
        \node at ( 4, 1.5) {$c_{3}$};
    \end{tikzpicture}}
    \hspace{2em}
    \vrule
    \hspace{2em}
    \scalebox{.75}{
    \begin{tikzpicture}
    [yscale=0.4,xscale=0.65, %
    voter/.style={anchor=south}]
    
        \foreach \i in {1,...,12}
    		\node[voter] at (\i-0.5, -1.5) {$\i$};
        
        \draw[fill=orange!\clrstr] (0, 0) rectangle (1, 1);
        \draw[fill=orange!\clrstr] (4, 0) rectangle (7, 1); 
        
        \draw[fill=red!40] (1, 1) rectangle (2, 2);
        \draw[fill=red!40] (4, 1) rectangle (6, 2);
        \draw[fill=red!40] (7, 1) rectangle (8, 2);
        
        \draw[fill=purple!35] (2, 2) rectangle (3, 3);
        \draw[fill=purple!35] (4, 2) rectangle (5, 3);
        \draw[fill=purple!35] (6, 2) rectangle (8, 3);
        
        \draw[fill=magenta!\clrstr] (3, 3) rectangle (4, 4);
        \draw[fill=magenta!\clrstr] (5, 3) rectangle (8, 4);
        
        \draw[fill=violet!40] (4, 4) rectangle (11,5);
        
        \draw[fill=blue!\clrstr] (4, 5) rectangle (10,6);
        \draw[fill=blue!\clrstr] (11,5) rectangle (12,6);
        
        \draw[fill=teal!\clrstr] (4, 6) rectangle (9, 7);
        \draw[fill=teal!\clrstr] (10,6) rectangle (12,7);
        
        \draw[fill=cyan!\clrstr] (4, 7) rectangle (8, 8);
        \draw[fill=cyan!\clrstr] (9, 7) rectangle (12,8);
        
        \draw[fill=green!\clrstr] (8, 0) rectangle (12,1);
        
        \draw[fill=lime!\clrstr] (8, 1) rectangle (12,2);
        
        \foreach \j in {1,...,4} {
            \node at (\j - 0.5, \j - 0.5) {$c_{\j}$}; }
        \foreach \j in {1,...,3} {
            \node at (4.5, \j - 0.5) {$c_{\j}$}; }
        \foreach \j in {2,...,4} {
            \node at (7.5, \j - 0.5) {$c_{\j}$}; }
        
        \foreach \j in {5,...,8} {
            \node at (5.5, \j - 0.5) {$c_{\j}$}; }
        \foreach \j in {6,...,8} {
            \node at (11.5, \j - 0.5) {$c_{\j}$}; }
        
        \foreach \j in {9, 10} {
            \node at (9.5, \j - 8.5) {$c_{\j}$}; }
    \end{tikzpicture}}
    \caption{
    Two profiles admitting IR committees that are not identified by common voting rules or proportionality axioms.
    }
    \label{fig:IF-consistency_counterexamples}
\end{figure}

\begin{example}\label{ex:threecandidates}
The first part of \Cref{fig:IF-consistency_counterexamples} shows an approval profile with $8$ voters and $3$ candidates. Assuming $k=2$, every voter $i \in N$ has $f_i = 1$. Thus, the only committee providing IR  is $\s = \quickset{c_1, c_2}$, which represents every voter once and, moreover, satisfies core stability. However, both $\s'=\{c_1,c_3\}$ and  $\s''=\{c_2,c_3\}$ are core stable as well (and in fact would be selected when choosing a committee maximizing the total number of approvals). %
Many common ABC voting rules would select either $\s'$ or $\s''$ (see \Cref{ssJRvsComMono}). One can argue that committee $\s$
is a ``fairer'' or ``more representative'' choice in this example.
\end{example}

\begin{example}\label{ex:tencandidates}
The second part of \Cref{fig:IF-consistency_counterexamples} shows an approval profile with $12$ voters and $10$ candidates. For $k=6$, we have $f_i = 1$ for $i \in \quickset{1, \ldots, 4}$ and $f_i = 2$ for $i \in \quickset{5, \ldots, 12}$. Here, the only committee providing IR is $\s = \quickset{c_1, c_2, c_3, c_4, c_9, c_{10}}$, representing each of the first four voters once, while representing all other voters at least twice. This committee is not core stable, because the group consisting of voters $5$ to $12$ would prefer $\{c_5,c_6,c_7,c_8\}$ to $\s$.
In order to appreciate the IR committee~$\s$, consider voters $1$ to $4$ and observe that these voters are completely ``symmetric.''  
Hence, from an ``equal treatment of equals'' perspective, 
if one of them is represented by an approved candidate in the committee, the same should hold for the others. In fact, the only core-stable committees that provide this kind of symmetry are $\{c_5,\dots,c_{10}\}$, in which one third of the electorate is not represented at all, or committees containing only two candidates among $c_5$ to $c_8$.
In the latter case,  by noticing that voters $9$ to $12$ are ``symmetric'' as well, we can argue similarly as above that they are not treated equally. 
Thus, the committee $\s$ that uniquely provides IR might be considered the ``fairest'' choice under an individualistic point of view.
\end{example}

The instance in \Cref{ex:tencandidates} shows that core stability and individual representation are incompatible in the strong sense that for this instance, the (nonempty) set of IR committees and the (nonempty) set of core-stable committees are disjoint. 

\begin{restatable}{proposition}{ircore}\label{thm:ircore}
IR is incompatible with core stability.
\end{restatable}

  \begin{figure}[t]
      \centering
      
    \scalebox{.9}{
      \begin{tikzpicture}
      [block/.style={minimum width = 2.4cm, minimum height = 0.5cm, rectangle, draw , font=\small},
      connec/.style={thick},
      notexist/.style={fill=gray!30},
      mightexist/.style={fill=gray!30, postaction={pattern={Lines[angle=-45, 
      distance=3mm,
      line width=10mm
      ]},pattern color=white,}}]
      
        \node[block, mightexist] (a) {core stability};
        \node[block, right = 0.8cm of a] (b) {EJR};
        \node[block, right = 0.8cm of b] (c) {PJR};
        \node[block, right = 0.8cm of c] (d) {JR};
        \node[block, notexist, above = 0.4cm of b] (e) {IR};
        \node[block, notexist, above = 0.4cm of d] (f) {semi-strong JR};
        \node[block, notexist, above = 0.4cm of f] (h) {PR};
         \draw[line, connec] (a.east)--(b.west);
        \draw[line, connec] (b.east)--(c.west);
        \draw[line, connec] (c.east)--(d.west);
        \draw[line, connec] (e.east)--(f.west);
        \draw[line, connec] (e.south)--(b.north);
        \draw[line, connec] (f.south)--(d.north);
        \draw[line, connec] (h.south)--(f.north);
        \draw[line, connec] (h.south west)--(c.north);
      \end{tikzpicture}
    }
     \caption{Relationships between different notions of representation. An arrow from $X$ to $Y$ signifies that $X$ implies $Y$. %
     A committee providing one of the shaded notions does not always exist (the case for core stability is an open problem). \new{PR is defined in \Cref{app:rel_perfect_repr}.}}
    \label{fig:relationship-notion-fairness}
  \end{figure}

Next, we show further incompatibility results for semi-strong JR. 

\begin{figure}
    \centering
        \scalebox{.75}{
    \begin{tikzpicture}
    [yscale=0.4,xscale=0.65,
    voter/.style={anchor=south}]
    
        \foreach \i in {1,...,8}
    		\node[voter] at (\i-0.5, -1.5) {$\i$};
        
        \draw[fill=orange!\clrstr] (0, 0) rectangle (4, 1);

        \draw[fill=red!\clrstr] (0, 1) rectangle (4, 2);
        
        \draw[fill=violet!\clrstr] (4, 0) rectangle (6, 1);
        
        \draw[fill=teal!\clrstr] (4, 1) rectangle (5, 2);
        \draw[fill=teal!\clrstr] (6, 1) rectangle (7, 2);

        \draw[fill=lime!\clrstr] (4, 2) rectangle (5, 3);
        \draw[fill=lime!\clrstr] (7, 2) rectangle (8, 3);
        
        \node at ( 2, 0.5) {$c_{1}$};
        \node at ( 2, 1.5) {$c_{2}$};
        \node at ( 5, 0.5) {$c_{3}$};
        \node at ( 4.5, 1.5) {$c_{4}$};
        \node at ( 6.5, 1.5) {$c_{4}$};
        \node at ( 4.5, 2.5) {$c_{5}$};
        \node at ( 7.5, 2.5) {$c_{5}$};
    \end{tikzpicture}}
    \caption{
    \new{Profile showing that semi-strong JR is incompatible with PJR, EJR, and corestability.
    } 
    }
    \label{fig:SSJR-counterexamples}
\end{figure}

\begin{restatable}{proposition}{ssjrejr}\label{thm:ssjrejr}
Semi-strong JR is incompatible with PJR, EJR, and core stability.
\end{restatable}

\vspace{-0.5cm}

\begin{proof}
Consider an ABC election with $n=8$, $k=4$ and the following approval profile: $A_1=\dots=A_4 = \quickset{c_1,c_2}$, $A_5 = \quickset{c_3,c_4,c_5}$, $A_6 = \quickset{c_3}$, $A_7 = \quickset{c_4}$, and $A_8 = \quickset{c_5}$. \new{For an illustration of this instance, see \Cref{fig:SSJR-counterexamples}.}
As $f_6=f_7=f_8=1$, every committee $\s$ that provides semi-strong JR must satisfy $\{c_3,c_4,c_5\} \subseteq \s$. But then we have \mbox{$\lvert W \cap (\cup_{i\in [4]}A_i)\rvert =1$}, even though the first four voters form a $2$-cohesive group.
As a consequence semi-strong JR is incompatible with PJR and EJR. 
Moreover, as in this instance the core is nonempty (e.g., the committee $\quickset{c_1,c_2,c_3,c_4}$ is core stable), we can also deduce that semi-strong JR is incompatible with core stability.
\end{proof}

In \Cref{app:rel_perfect_repr} we also establish the relation between \emph{perfect representation (PR)} as defined by \citet{SFF+17aNEW} and the two axioms we are interested in.
A graphical representation of the results of this section can be found in \Cref{fig:relationship-notion-fairness}.

\subsection{ABC Rules Violating IR}
\label{sec:ir_consistency}

Next, we consider the question whether we can find ABC rules that select IR committees \textit{whenever they exist}.  This question was already raised by \citet{ABC+16a} in the context of semi-strong JR, but remained open.
In other words, we look for rules that are ``consistent'' with individual representation. 

\begin{definition}[IR-consistency]
An ABC rule is \emph{consistent with individual representation}, or short \emph{IR-consistent}, if it outputs at least one IR committee for every ABC election that admits one.
\end{definition}

Consistency with semi-strong JR can be defined analogously. 
We show that all common ABC voting rules fail consistency with respect to both IR and semi-strong~JR.\footnote{Since neither IR nor semi-strong JR is always achievable (and semi-strong JR may be achievable in instances where IR is not) we can, in general, not deduce consistency regarding one of the notions from consistency regarding the other. However, all our examples in this section satisfy $f_i \leq 1$ for all voters $i$, such that semi-strong JR and IR coincide.}

\Cref{ex:threecandidates} already rules out any rule that always selects one of the candidates with the highest numbers of approvals, so-called \emph{approval winners}. In particular, this class of rules includes all common committee-monotonic ABC rules as well as other ``sequential'' rules like the Method of Equal Shares (MES), as these rules select one of the approval winners in the very first round.\footnote{For definitions of ABC rules not defined in this paper, we refer the reader to the survey by \citet{LaSk22a}.
\new{For a formal definition of sequentiality, see \citet{BDI+23a}.}
}

\begin{restatable}{proposition}{ssjrmono}\label{ssJRvsComMono}
No ABC voting rule that always selects one of the approval winners is IR-consistent.
\end{restatable}

Moreover, the rules PAV, Satisfaction-AV, and reverse-seqPAV select only committees including $c_3$ in \Cref{ex:threecandidates}, and thus fail IR-consistency as well. 
In \Cref{app:consistency} we provide additional examples showing that all remaining ABC rules mentioned in Table~4.1 of the survey by \citet{LaSk22a} fail IR-consistency as well.

\subsection{Computational Complexity}

Another open problem stated by \citet{ABC+16a} concerns the computational complexity of deciding whether a given ABC election admits a committee providing semi-strong~JR. We settle this question and the analogous one for individual representation by showing that both problems are \NP-hard. 

\begin{restatable}{theorem}{irnpcomplete}
It is \NP-hard to decide whether an ABC election admits an IR committee or a semi-strong JR committee.
\end{restatable}

\vspace{-0.5cm}

\begin{proof}
We reduce from \emph{exact cover by 3-sets}. Here, we are given a set of elements $X = \lbrace x_1, \dots, x_{3\ell} \rbrace$ and a collection $\mathcal T \subseteq 2^X$  of $3$-element subsets of $X$.  The goal is to find a partition of $X$ into sets from $\mathcal T$. The problem is \NP-hard even if each element appears in exactly three sets \cite{GaJo79a}.
We construct an ABC instance by setting $N=X$ and $\cC = \{c_i \mid T_i \in \mathcal{T}\}$, i.e., for each $T_i\in \mathcal{T}$ we have a candidate $c_i$.
 Further, for each set $T_i = \{x_{i_1},x_{i_2},x_{i_3}\}$  the candidate $c_i$ is approved exactly by voters $x_{i_1},x_{i_2}$, and $x_{i_3}$. We set $k = \ell$. Hence, only groups of $3$ voters corresponding to sets in $\mathcal T$ are 1-cohesive, and we get \new{$f_{x_i} \ge 1$} for each $x_i \in X$. 

Every exact cover by 3-sets corresponds to a committee of size $k$ where every voter is represented exactly once and thus provides IR in this instance. Conversely, every IR committee of the constructed ABC instance corresponds to a selection of sets from $\mathcal T$ such that every element in $X$ is covered exactly once.
Since $f_i = 1$ for every voter, the same argument holds for semi-strong JR as well.
\end{proof}

Moreover, it is hard to compute a voter's $f_i$-value.

\begin{restatable}{theorem}{finpcomplete}
Given an ABC instance, a voter $i \in N$, and $j \in \mathbb N$, it is \NP-complete to decide whether $f_i \ge j$ holds.
\end{restatable}

\vspace{-0.5cm}

\begin{proof}
It is easy to see that this problem is in \NP\, since any subset of voters including voter $i$ of size $j \cdot \frac{n}{k}$ and any subset of candidates of size $j$ approved by all selected voters serves as a witness. 

We reduce from \emph{balanced complete bipartite subgraph}. Here, we are given a bipartite graph $G = (V_1 \cup V_2, E)$ and an integer $j$ and the goal is to decide whether $G$ has $K_{j,j}$ as a subgraph, i.e., a subgraph consisting of $j$ vertices from $V_1$ and $j$ vertices from $V_2$ forming a bipartite clique. The problem is known to be \NP-hard \cite{GaJo79a}.
We construct an ABC instance by setting $N = V_1 \cup \{x\}$, $\cC = V_2 \cup \{y\}$ and $k = \lvert V_1 \rvert + 1$. Thus, $\frac{n}{k} 
= 1$. Each $v \in V_1$ approves exactly its neighbors in $G$, as well as $y$, while $x$ approves all candidates. It follows that $f_x \ge j+1$ if and only if there is a set of $j$ voters different from $x$ approving at least a common set of $j+1$ candidates. Since all voters approve $y$, this is equivalent to these $j$ voters all approving $j$ candidates different from $y$ and therefore by definition all being connected to these $j$ vertices in $V_2$. Thus, they form a $K_{j,j}$ if and only if $f_x \geq j + 1$.
\end{proof}

\section{Domain Restrictions}
\label{sec:dom_rest}

We have seen (\Cref{thm:no-beta-approx}) that non-trivial approximations of individual representation are impossible to obtain in general.
In this section, we explore whether this negative result can be circumvented by considering restricted domains of preferences. Domain restrictions for dichotomous (i.e., approval) preferences have been studied by \citet{ElLa15a} and \citet{yang2019tree}. 

Restricting attention to a well-structured domain often allows for axiomatic and algorithmic results that are not achievable otherwise \citep{ELP17a}.  
In the ABC setting, for example, it has recently been shown that a core-stable committee always exists in certain restricted domains \citep{PiSk22a}, whereas the existence of such committees is an open problem for the unrestricted domain.

We start by recalling the definitions of two classic restricted domains of dichotomous preferences: \textit{candidate interval} and \textit{voter interval} \citep{ElLa15a}.

\begin{definition}[Candidate Interval] An approval profile $A$ satisfies \emph{candidate interval (CI)} if there is a linear order over the candidates $C$ such that for every voter  $i \in N$, the approval set $A_i$ forms an interval of that order. 
\end{definition}

\begin{definition}[Voter Interval] An approval profile $A$ satisfies \emph{voter interval (VI)} if there is a linear order over the voters $N$ such that for every candidate $c_j \in C$, the set $N(\{c_j\})$ of voters approving $c_j$ forms an interval of that order. 
\end{definition}  
The profile in \Cref{ex:threecandidates} satisfies \textit{both} candidate interval and voter interval. In fact, a voter order witnessing VI is given in \Cref{fig:IF-consistency_counterexamples}. To see that the profile satisfies CI as well, consider the order $(c_1, c_3, c_2)$. The profile in \Cref{ex:tencandidates}, on the other hand, satisfies neither CI nor VI.

\citet{ElLa15a} have shown that it can be checked in polynomial time whether a profile satisfies CI or VI. (If the answer is yes, a linear order over candidates/voters can be found efficiently as well.)

Our first observation is that the \textit{candidate interval} domain is not helpful for our purposes: Indeed, the approval profile used to establish \Cref{thm:no-beta-approx} can easily be seen to satisfy CI. Thus, restricting preferences in this way does not yield any improved bounds.

\begin{corollary}
For every $k\ge 2$, there exists a CI profile $A$ such that $(A,k)$ does not admit an $(\alpha,\beta)$-IR committee with $\beta <k-1$ and any $\alpha \geq 1$.
\end{corollary}
 
Now, we turn our attention to the \textit{voter interval} domain. Due to the similarity between VI and CI, one might expect a similar result here. Surprisingly, however, we can prove a positive result for VI: We provide an algorithm that finds a $(2,4)$-IR committee in polynomial time for any VI profile. 

Before describing the high level idea of our algorithm, we state a useful property of VI profiles. 
Without loss of generality, we assume that the linear order witnessing VI is given by $(1,\dots,n)$. Moreover, for $a, b \in \mathbb{Z}$ with $a\le b$, we let $[a,b]$ denote the integer interval $\{a, a+1, \ldots, b\}$.
\begin{algorithm}[t]
\caption{$(2,4)$-IR for Voter Interval Profiles}\label{alg:ind-fair-DUE}
\begin{algorithmic}[1]
	\State $\s_0 \gets \emptyset$ \Comment{Round 1}
	\For {$i=1$ to $n$}
	    \State $S_i \gets \emptyset$
		\If{$|\s_{i-1} \cap A_i| < \floor{|N_{\geq i}|\cdot k/(2n) }$}
			 \State Let $S_i$ be an arbitrary subset of $S_i^*$ of size equal to $\floor{|N_{\geq i}|\cdot k/(2n)}-|\s_{i-1} \cap A_i|$ such that $|S_i\cap \s_{i-1} |$ is minimized
				
		\EndIf
		\State  $\s_i \gets \s_{i-1} \cup S_i$
	\EndFor		\label{algor:ind-fair-DUE-first-loop}
	\State $\hat{\s}_{0}  \gets \emptyset$ \Comment{Round 2}
	\For {$i=n$ to $1$}
	\State $S_i \gets \emptyset$
		\If{$|\hat{\s}_{n-i} \cap A_i| < \floor{|N_{< i}|\cdot k/(2n) }$}
		    \State Let $S_{i}$ be an arbitrary subset of $S_{i}^*$ of size equal to $\floor{|N_{< i}|\cdot k/(2n)}-|\hat{\s}_{n-i} \cap A_i|$ such that $|S_{i}\cap (\hat{\s}_{n-i}  \cup W_{n})|$ is minimized
		\EndIf
	\State  $\hat{\s}_{n-i+1} \gets \hat{\s}_{n-i} \cup S_{i}$
	\EndFor	\label{algor:ind-fair-DUE-second-loop}
	\State $S\gets$  an arbitrary subset of $\cC$ with $S\cap (W_n \cup \hat{\s}_{n})=\emptyset$ and $|S|=k-|W_n|-|\hat{W}_n|$
	\\ \Return $ \s_{n}  \cup \hat{\s}_{n} \cup S$
\end{algorithmic}
\end{algorithm}

\begin{restatable}{observation}{supportlem}
\label{lemm:supporters}
Let $i_1,i_2,i_3 \in [n]$ such that $i_1 < i_2 <i_3$. For any $S \subseteq \cC$, if $i_1 \in N(S)$ and $i_3 \in N(S)$, then $i_2 \in N(S)$.
\end{restatable}

Let $S_i^* \in \argmax_{S \subseteq  A_i} \{|S|\colon |N(S)|\geq |S| \cdot n/k \}$, i.e., a largest subset of $A_i$ approved by sufficiently many voters to validate the $f_i$-value. (If multiple such sets exist, we pick one of them arbitrarily.) From \Cref{lemm:supporters} we know that if $i_1$ and $i_2$ exist such that $i_1 < i_2 <i$ or $i < i_2 <i_1$, and if $i_1 \in N(S_i^*)$, then $i_2 \in N(S_i^*)$, i.e., $N(S_i^*)$ forms an interval of the order of voters including~$i$.  

Further, let  
$\N_{<i} \coloneqq \{i' \in N(S_i^*) \colon i' < i \}$ denote the set of voters in $N(S_i^*)$ that are ordered \textit{before} $i$ and let
$ \N_{\geq i}\coloneqq \{i' \in N(S_i^*) \colon i' \ge i\}$
denote the set of voters in $N(S_i^*)$ that are ordered \textit{after} $i$ (including $i$ itself).

\begin{observation}\label{Cor:shape-of-supp}
For each voter $i$, there exist $\xl \in [1,i]$  and $\xr \in [i,n]$ such that
$N_{<i}=[\xl,  i-1] \text{ and }
N_{\geq i }=[i,\xr]$. 
\end{observation}

Using this observation and the fact that $f_i=|S^*_i| \leq (|N_{< i}|+ |N_{\geq i}|)\cdot k/n $, \Cref{alg:ind-fair-DUE} returns a $(2,4)$-IR committee $\s$ for any VI profile as follows.  In the first round, iterating from voter $i=1$ to $n$, it selects at least $\floor{|N_{\geq i}|\cdot k/(2n) }$ candidates  that are approved by voter $i$. In the second round, iterating from voter $i=n$ to $1$, it selects at least $\floor{|N_{< i}|\cdot k/(2n) }$ candidates that are approved by voter $i$ (excluding the candidates that are selected in the first round). Together, this ensures $ |\s \cap A_i| \geq f_i/2-2$, where $\s$ is the set of selected candidates.

\begin{restatable}{theorem}{vi22aprox}\label{theor:IR_VI}
For every instance $(A,k)$ such that $A$ satisfies voter interval, 
\Cref{alg:ind-fair-DUE} returns a $(2,4)$-IR committee in polynomial time.
\end{restatable}

\vspace{-0.5cm}

\begin{proof}
Let $\s=\s_{n} \cup \hat{\s}_{n} \cup S$ be the committee returned by \Cref{alg:ind-fair-DUE}, and let  $f_{\geq i}=  |N_{\geq i}|\cdot k/n$ and $f_{<i}= |N_{<i}|\cdot k/n$. In the first round we ensure that $|W_{n} \cap A_i| \geq \floor{f_{\geq i}/2} $, as at iteration $i$ if $|W_{i-1}\cap A_i|< \floor{f_{\geq i}/2} $, we include $ \floor{f_{\geq i}/2}-|W_{i-1}\cap A_i|$ candidates into $\s_i$ that are not already included. Similarly, in the second round we ensure that   $|\hat{\s}_{n} \cap A_i| \geq \floor{f_{< i}/2} $   as at iteration $n-i+1$ if $|\hat{\s}_{n-i}\cap A_i|< \floor{f_{< i}/2} $, we include $ \floor{f_{< i}/2}-|\hat{\s}_{n-i}\cap A_i|$ candidates into $\hat{W}_{n-i+1}$ that are not already included.
As  $f_i \leq f_{\geq i}+f_{<i}$, for each $i \in N$ we have that
\[|\s \cap A_i| \geq  \floor{f_{\geq i}/2}+\floor{f_{<i}/2} 
\geq f_i/2-2 \text,\]
and therefore $2 \cdot |\s \cap A_i| + 4 \ge f_i $. Thus, we conclude that~$\s$ provides $(2,4)$-IR.

\medskip

Now we show that $|\s_{n}|\leq k/2$ and $|\hat{\s}_n| < k/2$. We first consider $\s_n$.
\begin{restatable}{lemma}{lessk22lemmaround1}
\label{lem:IR-VI-1}
$|\s_i|\leq \frac{((i-1)+ |N_{\geq i}|)\cdot k}{2n}$ for all $i \in [n]$. 
\end{restatable}

\vspace{-0.5cm}

\begin{proof}
\renewcommand{\qedsymbol}{\diamond}
We prove the lemma using induction.
For $i=1$, $\s_1=\floor{f_{\geq 1}/2}\leq \frac{ |N_{\geq 1}|\cdot k}{2n} $ and the statement holds.
Assume that for all $ t' <t$, we have 
$|\s_{t'}|\leq  \frac{((t'-1)+ |N_{\ge t'}|)\cdot k}{2n}$.

We show that the statement holds for $\s_t$.
Note that 
\begin{align}
   |\s_{t}|= |\s_{t-1}|+ \floor{f_{\geq t}/2}-|\s_{t-1}\cap A_t| \label{eq:VI-t-1}
\end{align}
as at iteration $t$, the algorithm  adds $\floor{f_{\geq t}/2}-|\s_{t-1}\cap A_t|$ candidates to $\s_{t}$.
\new{Let $t^*=\max\{r \in \{0,\ldots, t-1\}\colon t \in N(S^*_{t-r}) \}$. 
In words, $t^*$ denotes the leftmost voter in the linear order $(1,\ldots, n)$ such that  $t$ approves $N(S^*_{t^*})$. Note that $[t-t^*, t]\subseteq N(S^*_{t^*})$. }

\new{First,  assume that  $t^*=0$. This means that $t$ does not approve any    $N(S^*_{t''})$ for $t''<t$. From this we get that  $|N_{\ge t-1}|=1$, since  $N_{\ge t-1}$ is an interval that contains agents to the right side of $t-1$ (including $t-1$), but since $t$ is not part of this interval, no agent to the right of $t$ can be part of it either. Then, 
from \Cref{eq:VI-t-1} we have that}

\begin{align*}
   |\s_{t}|\leq&  \frac{((t-2)+ |N_{\ge t-1}|)\cdot k}{2n} 
    + \floor{f_{\geq t}/2}-|\s_{t-1}\cap A_t| \\
   \leq &\frac{((t-2)+ 1)\cdot k}{2n}+ \floor{f_{\geq t}/2}-|\s_{t-1}\cap A_t|\\
   \leq &\frac{(t-1+ |N_{\geq t}|  )\cdot k}{2n}.
\end{align*}

Now assume that $t^*>0$.
First, using induction, we show that 
\begin{align*}
|\s_{t}|= |\s_{t-r}|+ \floor{f_{\geq t}/2}- |\s_{t-r} \cap A_t|
\end{align*}
for every $r \in [1,t^*+1]$. \new{Intuitively, this follows from the fact that all the candidates that are  added during iterations from $t -t^*$ up to $t-1$ are approved by voter $t$ as well, since $t\in N(S^*_{t-r})$ for all $r \in [1,t^*]$.} For $r=1$, the claim immediately  follows from \Cref{eq:VI-t-1}. Assume that for all $q'<q$ it holds that $|\s_{t}|= |\s_{t-q'}|+ \floor{f_{\geq t}/2}- |\s_{t-q'} \cap A_t|$.
We have 
\begin{align}
    |\s_{t-(q-1)} \cap A_{t}|= &\floor{f_{\geq t-(q-1)}/2} - |\s_{t-q} \cap A_{t-(q-1)}| 
    + |\s_{t-q}\cap A_t|\label{eq:VI-induction}
\end{align} 
as at iteration $t-(q-1)$ we add $\floor{f_{\geq t-(q-1)}/2}-{|\s_{t-q} \cap A_{t-(q-1)}|}$ candidates from $S^*_{t-(q-1)}$ to $\s_{t-(q-1)}$, and as $t \in N(S^*_{t-(q-1)})$, these candidates are approved by $t$, too. Then, 
\begin{align*}
|\s_{t}|&=| \s_{t-(q-1)}| + \floor{f_{\geq t}/2}-|\s_{t-(q-1)}\cap A_t| \\
&= |\s_{t-q}|+ \floor{f_{\geq t-(q-1)} /2}-|\s_{t-q} \cap A_{t-(q-1)}| + \floor{f_{\geq t}/2}-|\s_{t-(q-1)}\cap A_t| \\
&= |\s_{t-q}|+ \floor{f_{\geq t}/2}-|\s_{t-q}\cap A_t|,
\end{align*}
where the second transition follows since at iteration $t-(q-1)$, $\floor{f_{\geq t-(q-1)}/2}-|\s_{t-q} \cap A_{t-(q-1)}|$ candidates are added to $W_{t-(q-1)}$, and 
the third transition follows from \Cref{eq:VI-induction}.

Now, we distinguish two cases.

\smallskip

\noindent \textbf{Case 1:} $t^*=t-1$.
Here, we have 
\begin{align*}
    |\s_t|&=|\s_{t-(t^*+1)}|+\floor{f_{\geq t/2}}-|\s_{t-(t^*+1)} \cap A_t|\\
    &=|\s_{0}|+\floor{f_{\geq t/2}}-|\s_{0} \cap A_t|
    =\frac{\lvert N_{\geq t} \rvert \cdot k}{2n}.
\end{align*}

\noindent \textbf{Case 2:} \new{$t^*\geq 1$}.
Here, we have 
\begin{align*}
    |\s_t| &= |\s_{t-(t^*+1)}|+\floor{f_{\geq t/2}}-|\s_{t-(t^*+1)} \cap A_t|\\
    &\leq \frac{(t-(t^*+1)-1+ |N_{\geq t-(t^*+1) }|  )\cdot k}{2n} +\frac{|N_{\geq t}|\cdot k}{2n}
    \leq \frac{(t-1+ |N_{\geq t}|  )\cdot k}{2n},
\end{align*}
where the third inequality follows from the fact that $|N_{\geq t-(t^*+1) }| \leq t-(t-(t^*+1))$, as $t$ is not in $N(S^*_{t-(t^*+1)})$.
\end{proof}
As Rounds 1 and 2 of \Cref{alg:ind-fair-DUE} are symmetric, with similar arguments, we can show the following lemma. \new{The proof can be found in \Cref{app:proof_lessklemmaroundtwo}.}
\begin{restatable}{lemma}{lessklemmaroundtwo}
\label{lem:IR-VI-2}
$|\hat{\s_i}|  \leq  \frac{((i-1)+ |N_{< n-i+1}|)\cdot k}{2n}$
for all $i \in [n]$.
\end{restatable}

From \Cref{lem:IR-VI-1}, for $i=n$, we have $|N_{\ge n}|\leq 1$, and hence
$|\s_{n}|\leq \frac{((n-1)+1)\cdot k}{2n}\leq k/2$.
From \Cref{lem:IR-VI-2}, for $i=n$, we have $|N_{< 1}|=0$, and hence
$|\hat{\s_{n}}|\leq \frac{((n-1)+0)\cdot k}{2n}< k/2$.
Thus, $|W|=|\s_{n}|+|\hat{\s}_{n}|<k$.

\smallskip

Lastly, we show that $S^*_i$, and thus $f_i$, can be computed in polynomial time.
For this, we employ \Cref{Cor:shape-of-supp}, i.e., the fact that $N(S^*_i)$ forms an interval of voters that includes $i$. We consider all such intervals and for each of them calculate the maximum subset of candidates that the voters in this interval deserve due to their size.

\begin{algorithm}[t]
	\caption{Finding $f_i$ and $S^*_i$}\label{alg:find-fi}
	\begin{algorithmic}[1]
	    \State $f_i\gets 0$
	    \State $S^*_i \gets \emptyset$
        \For {$\xl=1$ to $i$}
			\For {$\xr=i$ to $n$}
			    \State $S \gets \emptyset$
			    \For {$j=1$ to $m$}
			        \If{$  \{\xl,\dots,\xr\} \subseteq N(\{c_j\}$) }		
			            \State $S \gets S \cup \{c_j\}$	
			        \EndIf
			    \EndFor
			    \State $\ell^*\gets \argmax_{\ell \in \mathbb{N}}\{\ell: \xr-\xl+1\geq \ell \cdot n/k\}$

			    \If {$\min\{\ell^*, |S|\}>f_i$}
                    \State $S^*_i\gets$ an arbitrary subset of $S$ of size $\min\{\ell^*, |S|\}$
                      \State $f_i\gets |S^*_i|$
			    \EndIf
			\EndFor
		\EndFor
	\\ \Return $f_i$, $S^*_i$
	\end{algorithmic}
\end{algorithm}

\begin{restatable}{lemma}{findfi}\label{lemm:find-fi}
For any voter $i$, $f_i$, and  $S^*_i$ can be computed in polynomial time. 
\end{restatable}

\vspace{-0.5cm}

\begin{proof}
\renewcommand{\qedsymbol}{\diamond}
We show that \Cref{alg:find-fi} correctly computes both $f_i$ and $S^*_i$. 
For each interval $[\xl,\xr]$ where $\xl\leq i$ and $i \leq \xr$, the algorithm finds the maximum number of candidates that this interval is eligible to elect, denoted by $\ell^*$.  Moreover, $S \subseteq \cC$ denotes the set of candidates that is approved by all the voters in the interval. The algorithm calculates the  maximum subset of  $S$ that can be elected by the voters in the interval as $\min\{\ell^*,|S|\}$.  Then, $S^*_i$ and $f_i$  are updated properly by assigning them the biggest subset and the size of it, respectively, that an interval of voters can elect.   

Assume for contradiction that the algorithm returns a subset $S'$ with $|S'|< |S^*_i|$. This means that there is an interval  $[\xl,\xr]$ of voters that can elect  $S^*_i$. Thus, when the algorithm considers this interval, it would return a subset of size at least $|S^*_i|$, a contradiction. 
\end{proof}
\noindent
This concludes the proof of 
\Cref{theor:IR_VI}.
\end{proof}

Further, we can show that the bound provided by \Cref{theor:IR_VI} is almost tight up to the additive part of $4$.

\begin{restatable}{theorem}{VIlowerbound}\label{lemm:VI-lower-bound}
For every $k\ge 3$, there exists a VI profile $A$ such that $(A,k)$ that does not admit an $(\alpha,0)$-IR committee with $\alpha<2-2/k$.
\end{restatable}

\vspace{-0.5cm}

\begin{proof}
Fix $k\ge 3$ and consider the following instance with $n\geq k$ and $m =2(k-1)$.  
All voters \mbox{$i \in [2,n-1]$} approve all the candidates, while $A_1=\{c_1,\dots,c_{k-1}\}$ and $A_n=\{c_{k},\dots,c_m\} $. 
Notice that this profile is VI. Indeed, if we order the voters as $1,2,\dots, n$, then the voters that approve each candidate  form an interval  of the ordering. Now, we  see that  $f_1=f_n=k-1$, but  for each $\s \subseteq \cC$ with $|\s| \leq k$, 
 either $|A_1 \cap \s|\leq k/2$ or $|A_n \cap \s|\leq k/2$.
\end{proof}

In \Cref{app:inapr-rule-VI}, we show that all common ABC rules may fail to return a committee that provides $(2,4)$-IR for VI preferences. 

Beyond VI and CI, many other domain restrictions have been studied in the literature. In \Cref{app:fur_dom_res}, we provide lower and upper bounds for $(\alpha, \beta)$-IR  for all domain restrictions introduced by \citet{ElLa15a} and \citet{yang2019tree}.
Any domain that is more restrictive than VI
inherits the guarantee of a $(2,4)$-IR committee from VI\,---\,but we show that in some cases we can achieve better approximation guarantees. On the other hand, any domain that is more general than CI inherits the inapproximability from CI. In fact, we show that the same lower bound applies even in a slightly more restricted domain introduced by  \citet{yang2019tree}. 
Moreover, we show that committees satisfying IR (without approximation) always exist and can be found in polynomial time for a subclass of VI profiles.
We also determined for which of the considered domain restrictions a semi-strong JR committee is guaranteed to exist. 
For a summary of our results, see \Cref{tab:IF_guarantees_on_structured_profiles} in \Cref{app:fur_dom_res}.

\section{Experimental Results}\label{sec:experiments}

To complement our theoretical results, we performed experiments on generated approval profiles in order to check how often IR committees exist and how often they are selected by common ABC rules.

\subsection{Setup}

\new{Inspired by \citet{PPSS21a} and \citet{SFJ+22a}, we used five models to generate approval profiles:
a voter-interval model (VI), a candidate-interval model (CI), an impartial culture model (IC), the truncated urn model (Truncated), and the resampling model (Resampling). 
All generated approval profiles have $100$ voters and $50$ candidates. For each of the five models, we generated $1000$ profiles, using a variety of parameters. For each generated profile, we created 50 ABC elections, one for each $k \in [50]$. Thus, the total number of generated ABC elections is 250,000.} 

Our first model is the \emph{voter interval Euclidean model} (VI). Here, we choose a location in $[0,1]$ uniformly at random for each voter and candidate. Further, for each candidate we choose a radius according to \new{$\lvert \mathcal{N}(0, r) \rvert$ for a parameter $r$}. A candidate is approved by all voters in its radius. \new{We select $100$ instances for each $r \in \{\frac{\ell}{20}\colon \ell \in [1, 10]\}$.}

Our second model is the \emph{candidate interval Euclidean model} (CI). Here, we again choose a location in $[0,1]$ uniformly at random for each voter and candidate as well as a radius according to \new{$\lvert \mathcal{N}(0, r) \rvert$ for a parameter $r$} for each voter. A voter approves all candidates in its radius. \new{We select $100$ instances for each $r \in \{\frac{\ell}{20}\colon \ell \in [1, 10]\}$.}

In the \emph{impartial culture model} (IC), for all voters and each candidate, the candidate is approved by the voter with probability \new{$r$}.  \new{We select $100$ instances for each $r \in \{\frac{\ell}{20}\colon \ell \in [1, 10]\}$.}

\new{Finally, for the \emph{truncated urn model} and the \emph{resampling model}, we follow the approach of \citet{SFJ+22a}, who use these (and other) models to draw ``maps of elections.'' To cover a wide variety of locations on those maps, we pick several different parameter combinations for those models.\footnote{
\new{For the truncated urn model, which uses parameters $(\alpha,p) \in [0,1]^2$, we use the following 10 combinations of parameters:
$(0.1,0.5)$, $(0.9,0.5)$, $(0.3,0.3), (0.3,0.5), (0.3,0.7)$, $(0.5,0.3), (0.5,0.5), (0.5,0.7)$, $(0.7,0.4)$, and  $(0.7,0.6)$.
For the resampling model with parameters $(\phi,p) \in [0,1]^2$, we use the same set of parameter combinations.} 
}
}

We consider the ABC rules AV, PAV, seq-PAV, Greedy Monroe, MES, seq-Phragmén, and sequential Chamberlin--Courant (seq-CC).

\begin{figure*}[t]
\begin{center}
   \includegraphics[scale = 0.6]{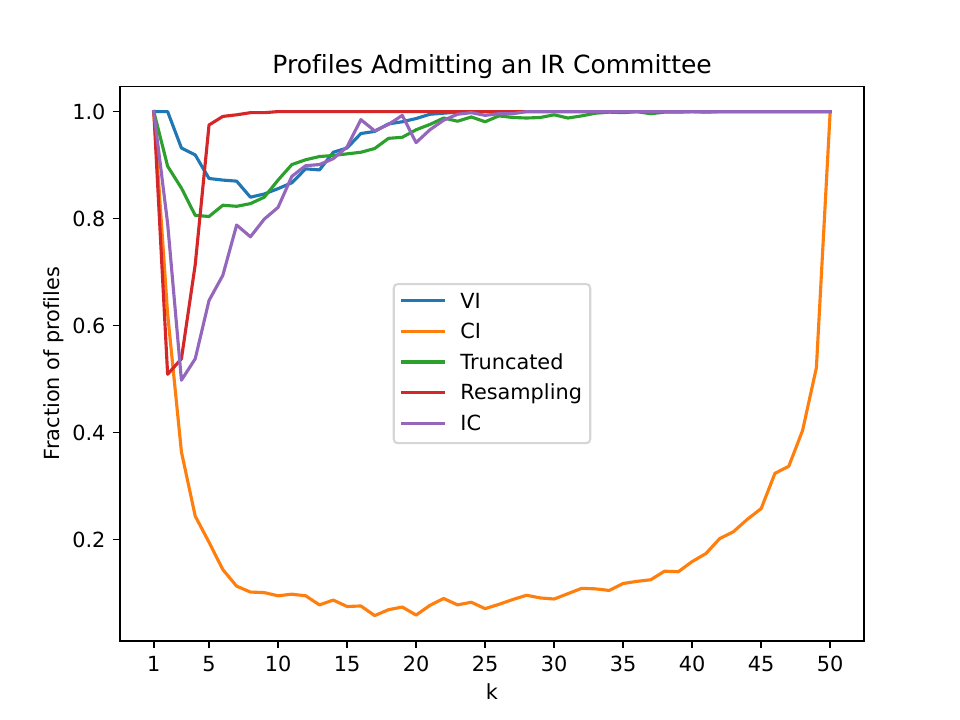}   
\end{center}
     \caption{The ratio of generated profiles that admit an IR committee.}
     \label{fig:experiements1}
 \end{figure*}

\subsection{Results}
First, we studied how often the generated approval profiles admit an IR committee. The results are shown in \Cref{fig:experiements1}. We found that IR committees exist quite often, especially for larger values of $k$. \new{In particular, profiles generated by the VI model or by the truncated urn model admit IR committees in more than 80\% of instances, for all values of $k$.} On the other hand, profiles generated by the CI model \new{rarely} admit IR committees. This striking contrast between VI and CI, which is reminiscent of our theoretical results in \Cref{sec:dom_rest}, can be explained with a feature of the preference generation model: 
Due to the way we generate CI preferences, 
\new{many voters tend to have rather large approval sets.} 
These voters approving many candidates are then part of multiple cohesive groups, not all of which can be represented in an IR manner. (A similar situation can be observed in the profile constructed in the proof of \Cref{thm:no-beta-approx}.)

\begin{figure*}[t]
\begin{center}
   \includegraphics[scale = 0.6]{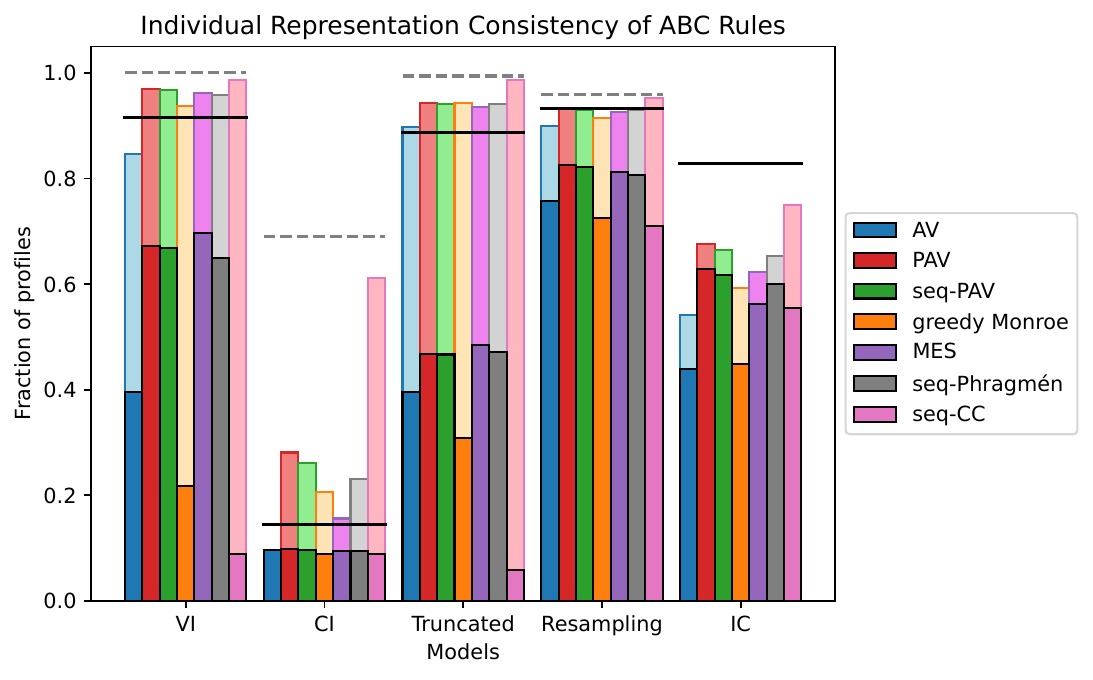}
\end{center}
     \caption{For each model and each voting rule, the bold colored part of the bar indicates the ratio of instances the rule returned an IR committee, while the pale-colored part indicates the same ratio for semi-strong JR, averaged over all values $k$ with $2 \le k \le 20$. For each model, the black line indicates the fraction of instances admitting an IR committee, while the gray dashed line indicates the ratio of instances admitting a semi-strong JR committee.}
     \label{fig:experiements2} 
 \end{figure*} 
 
Second, we studied how often different ABC rules select a committee providing IR (or semi-strong JR). In order not to dilute our results, we restricted $k$ to the ``interesting'' range between $2$ and $20$. The results are shown in \Cref{fig:experiements2}. Of course, the fraction of profiles for which a rule selects an IR (or semi-strong JR) committee is upper-bounded by the fraction of profiles that admit such a committee. For each model, the latter fraction is depicted in the graph as a solid black line for IR, and a dashed gray line for semi-strong JR. 
While no rule manages to find an IR committee every time one exists, the rules PAV, sequential PAV, MES, and sequential Phragmén select IR committees often. 
For the small fraction of CI profiles that admit an IR committee, all considered rules do a good job in finding one.  
Since seq-CC greedily optimizes the amount of voters that are represented at least once, it finds a committee providing semi-strong JR in almost all profiles that admit one. But as the rule does not aim at representing voters more than once, it rarely produces IR committees. 
In the profiles generated by the IC model, IR often coincides with semi-strong JR (for $k\le 20$) because almost all non-zero $f_i$-values are $1$. This is in line with the effect noticed by \citet{BFKN19}, whose experiments showed that EJR and JR are very likely to coincide under IC.

\section{Discussion}

Based on the observations that 
common axioms in approval-based committee voting do not address the representation of individual voters, and that 
common voting rules sometimes unfairly distinguish between such voters, we formalize \textit{individual representation} (IR) as a requirement for committees. 
We find that all common voting rules fail to select IR committees, even when these exist.
Nevertheless, for some restricted domains\,---\,most prominently, voter interval preferences\,---\,we provide polynomial-time algorithms for finding committees that provide a good approximation to IR. 
Our experimental results suggest that IR is achievable in many instances that follow somewhat realistic preferences.
It remains an open problem to find intuitive voting rules that provide (approximate) IR whenever possible.

\section*{Acknowledgments}

A short version of this paper appears in the \textit{Proceedings of the 36th AAAI Conference on Artificial Intelligence (AAAI 2022)}.
We would like to thank Nisarg Shah and the anonymous reviewers at AAAI and Social Choice and Welfare for their helpful comments.
This work was supported by the Deutsche Forschungsgemeinschaft under grant BR~4744/2-1 and the Graduiertenkolleg ``Facets of Complexity'' (GRK~2434).

\bibliographystyle{plainnat}
\bibliography{bibliography,abb,algo}

\clearpage

\appendix

\section{Omitted Proofs}\label{app:missed-proofs}
\subsection{Proof of \Cref{lem:IR-VI-2}}\label{app:proof_lessklemmaroundtwo}
\lessklemmaroundtwo*

\vspace{-0.5cm}

\begin{proof}
We prove the lemma using induction.
For $i=1$, we have that $\hat{\s}_1=\floor{f_{<n }/2}\leq \frac{ |N_{< n}|\cdot k}{2n} $ and the statement holds.
Assume that, for all $ t' <t$, we have 
$|\hat{\s}_{t'}|\leq  \frac{((t'-1)+ |N_{<n- t'+1}|)\cdot k}{2n}$.

We show that the statement holds for $\hat{\s}_t$.
Note that 
\begin{align}
   |\hat{\s}_{t}|= |\hat{\s}_{t-1}|+ \floor{f_{<n- t+1}/2}-|\hat{\s}_{t-1}\cap A_{n-t+1}| \label{eq:VI-t-1-app}
\end{align}
as at iteration $t$, the algorithm  adds $ \floor{f_{<n- t+1}/2}-|\hat{\s}_{t-1}\cap A_{n-t+1}|$ candidates to $\hat{\s}_{t}$.
Let $t^*=\max\{r \in \{0, \ldots, t-1\}\colon n-t+1 \in N(S^*_{n-t+1+r}) \}$.  
If $t^*=0$, from \Cref{eq:VI-t-1-app} we get that
\begin{align*}
   |\hat{\s}_{t}|\leq&  \frac{((t-2)+ |N_{<n- t+2}|)\cdot k}{2n} + \floor{f_{<n- t+1}/2}-|\hat{\s}_{t-1}\cap A_{n-t+1}|
   \\
   \leq &\frac{((t-2)+ 0)\cdot k}{2n} + \floor{f_{<n- t+1}/2}-|\hat{\s}_{t-1}\cap A_{n-t+1}|\\
   \leq &\frac{((t-2)+ |N_{< n- t+1}|  )\cdot k}{2n},
\end{align*} 
where the second inequality follows from the fact that $|N_{< n-t+2}|=0$, as $n-t+1 \notin N(S^*_{n-t+2})$. Now assume that $t^*>0$.
First, using induction, we show that 
\begin{align*}
|\hat{\s}_{t}|= |\hat{\s}_{t-r}|+ \floor{f_{< n-t+1}/2}- |\hat{\s}_{t-r} \cap A_{n-t+1}|
\end{align*}
for any $r \in [1,t^*+1]$. 

For $r=1$, the claim immediately  follows from \Cref{eq:VI-t-1-app}. Assume that for all $q'<q$ it holds that $|\hat{\s}_{t}|= |\hat{\s}_{t-q'}|+ \floor{f_{<n-t+1}/2}- |\hat{\s}_{t-q'} \cap A_{n-t+1}|$.
Now, we have 
\begin{align}
    |\hat{\s}_{t-(q-1)} \cap A_{n-t+1}| \nonumber 
    = &\floor{f_{< n-t+(q-1)+1}/2} 
    - |\hat{\s}_{t-q} \cap A_{n-t+(q-1)+1}| \\
    &+ |\hat{\s}_{t-q}\cap A_{n-t+1}|\label{eq:VI-induction-app}
\end{align} 
as at iteration $t-(q-1)$ we add$\floor{f_{< n-t+(q-1)+1}/2}   - |\hat{\s}_{t-q} \cap A_{n-t+(q-1)+1}|$ candidates from $S^*_{n-t+(q-1)+1}$ to $\hat{\s}_{t-(q-1)}$, and as $n-t+1 \in N(S^*_{n-t+(q-1)+1})$, these candidates are approved by $n-t+1$, too. Then, 
\begin{align*}
|\hat{\s}_{t}|&=| \hat{\s}_{t-(q-1)}| + \floor{f_{<n-t+1}/2} -|\hat{\s}_{t-(q-1)}\cap A_{n-t+1}| \\
&= |\hat{\s}_{t-q}|+ \floor{f_{< n-t+(q-1)+1}/2} - |\hat{\s}_{t-q} \cap A_{n-t+(q-1)+1}| \\
&\phantom{= |\hat{\s}_{t-q}|} \,\, + \floor{f_{<n-t+1}/2}-|\hat{\s}_{t-(q-1)}\cap A_{n-t+1}| \\
&= |\hat{\s}_{t-q}|+ \floor{f_{<n-t+1}/2}-|\hat{\s}_{t-q}\cap A_{n-t+1}|,
\end{align*}
where the second transition follows since at iteration $t-(q-1)$, $\floor{f_{< n-t+(q-1)+1}/2}   - |\hat{\s}_{t-q} \cap A_{n-t+(q-1)+1}|$ candidates are added to $\hat{\s}_{t-(q-1)}$, and 
the third transition follows from \Cref{eq:VI-induction-app}.

Now, we distinguish two cases.

\smallskip

\noindent \textbf{Case 1:} $t^*=t-1$
Here, we have 
\begin{align*}
    |\hat{\s}_t|&=|\hat{\s}_{t-(t^*+1)}|+\floor{f_{<n-t+1 /2}} -|\hat{\s}_{t-(t^*+1)} \cap A_{n-t+1}|\\
    &=|\hat{\s}_{0}|+\floor{f_{<n-t+1 }/2}-|\hat{\s}_{0} \cap A_{n-t+1}| =\frac{\lvert N_{<n- t+1} \rvert \cdot k}{2n}.
\end{align*}

\noindent \textbf{Case 2:} $t^*\geq 1$
Here, we have 
\begin{align*}
    |\hat{\s}_t| &= |\hat{\s}_{t-(t^*+1)}|+\floor{f_{<n-t+1}/2} - |\hat{\s}_{t-(t^*+1)} \cap A_{n-t+1}|\\
    &\leq \frac{(t-(t^*+1)-1+ |N_{<n- t+(t^*+1)+1 }|  )\cdot k}{2n} + \frac{|N_{<n-t+1}|\cdot k}{2n}\\
    &\leq \frac{(t-2+ |N_{<n- t+1}|  )\cdot k}{2n},
\end{align*}
where the third inequality follows from the fact that $|N_{<n- t+(t^*+1)+1 }| \leq t^*$, as $n-t+1$ is not in $N(S^*_{n-t+(t^*+1)+1})$.
\end{proof}

\section{Relation to Perfect Representation}\label{app:rel_perfect_repr}

Consider an approval profile $A$ and a committee size $k$ such that $k$ divides the number of voters $n$. A committee $W$ of size $k$ provides \emph{perfect representation (PR)} \citep{SFF+17aNEW} if it is possible to partition the electorate $\N$ into $k$ pairwise disjoint subsets $\N_1, \ldots, \N_k$ of size $\frac{n}{k}$ each, and assign a distinct candidate from $W$ to each of the subsets in such a way that for each subset all the voters in the subset approve of the assigned candidate.

It is known that not all ABC voting instances where $k$ divides $n$ admit a committee providing PR \cite{SFF+17aNEW}. Thus, in the following, we will call an ABC election a \emph{PR-instance} if it admits a PR committee. 

\begin{restatable}{proposition}{prssjr}
On PR-instances, every committee providing perfect representation also provides semi-strong JR, but not the other way around.
\label{thm:ssJRvsPR}
\end{restatable}

\vspace{-0.5cm}

\begin{proof}
Assume a PR-instance is given together with a committee $\s$ that provides perfect representation. By definition, $|A_i \cap \s| \geq 1$ for every voter $i \in \N$ . Thus, $W$ provides semi-strong JR.
As a counterexample for the other direction, consider the following instance with $n=6$ and $k=3$.
\[   A_1 = A_2 = \quickset{c_1, c_4}, \quad
A_3 = \quickset{c_2,c_4,c_6}, \quad
A_4 = \quickset{c_2,c_5}, \quad
A_5 = A_6 = \quickset{c_3,c_5}  \]
The committee $\s = \quickset{c_4,c_5,c_6}$ provides semi-strong JR but not perfect representation (whereas $\quickset{c_1,c_2,c_3}$ provides both).
\end{proof}

Note that \citet[Theorem~4]{SFF+17aNEW} establish that on PR-instances, perfect representation also implies PJR. It follows that, whereas semi-strong JR and PJR are incompatible in general (as we proved in \Cref{sec:ind_rep}), every PR-instance admits a committee that satisfies semi-strong JR, PJR, and perfect representation.

\begin{restatable}{proposition}{pr_ir}
There are PR-instances that do not admit an IR committee. Moreover, there are PR-instances where an IR committee exists but does not provide perfect representation.
\end{restatable}

\vspace{-0.5cm}

\begin{proof}
Regarding the first claim consider the following instance with $n=8$ and $k=4$. 
\begin{align*}
    A_1 &= \quickset{c_1} &&A_5 = \quickset{c_1,c_5,c_6} \\
    A_2 &= \quickset{c_2} &&A_6 = \quickset{c_2,c_5,c_6} \\
    A_3 &= \quickset{c_3} &&A_7 = \quickset{c_3,c_5,c_6} \\
    A_4 &= \quickset{c_4} &&A_8 = \quickset{c_4,c_5,c_6} 
\end{align*}
Here $\quickset{c_1, \ldots, c_4}$ provides perfect representation (and semi-strong JR) but not IR (which is not achievable in this instance).
Regarding the second claim, again consider the instance from the proof of \Cref{thm:ssJRvsPR} with $n=6$ and $k=3$. Here, the committee $\quickset{c_4,c_5,c_6}$ provides IR, but not perfect representation.
\end{proof}

\section{Counterexamples for IR-Consistency}
\label{app:consistency}

Here we provide further examples showing that all common ABC rules violate IR-consistency. Note that all committee monotone rules (including AV, SAV, sequential PAV, and sequential Phragmén) as well as MES and PAV were already ruled out to satisfy IR-consistency in \Cref{sec:ir_consistency}.

\begin{example}
Consider the following profile with $n=16$ voters and assume $k = 4$:
\begin{align*}
    A_1 &= A_2 = A_3 = \quickset{ c_1 }  &&A_4 = \quickset{ c_1, c_5 } \\
    A_5 &= A_6 = A_7 = \quickset{ c_2 }  &&A_8 = \quickset{ c_2, c_5 } \\
    A_9 &= A_{10} = A_{11} = \quickset{ c_3 }  &&A_{12} = \quickset{ c_3, c_5 } \\
    A_{13} &= A_{14} = A_{15} = \quickset{ c_4 }  &&A_{16} = \quickset{ c_5 }.
\end{align*}
Here we have $f_i = 1$ for all voters $i \in N \setminus \quickset{13, 14, 15}$ and thus the only committee providing individual representation is $\s = \{ c_1,c_2,c_3,c_5 \}$. Chamberlin-Courant-AV (CCAV), Monroe-AV and PAV with the weight-vector $(1, \frac{1}{n}, \frac{1}{n^2}, \ldots )$, 
which provide an IR committee in the example of \Cref{ssJRvsComMono}, choose $\{ c_1,c_2,c_3,c_4 \}$ and thus fail individual representation.
\end{example}

The only two remaining rules from Table~4.1 in the survey by \citet{LaSk22a} are leximin-Phragmén \citep{BFJL24a} (referred to as ``leximax-Phragmén'' by \citet{LaSk22a}) and Minimax-AV \citep{BKS07a}.
\begin{example}
Consider the following profile with $n=6$ voters and let $k=3$:
    \begin{align*}
        A_1 &= \quickset{c_1} &&A_3 = \quickset{c_1, \ldots, c_5} &&A_5 = \quickset{c_4,c_5,c_6} \\
        A_2 &= \quickset{c_2} &&A_4 = \quickset{c_3,\ldots,c_6}   &&A_6 = \quickset{c_5,c_6}.
    \end{align*}
Here we have $f_i = 1$ for all voters $i \in N$ and thus $\quickset{c_1,c_2,c_5}$ and $\quickset{c_1,c_2,c_6}$ are the only committees providing individual representation (or semi-strong JR). It turns out that leximin-Phragm{\'e}n does not select any of these two. It is easy to see that the load distribution with the minimum maximal load for either of the two IR committees is $(\frac{2}{3}, \frac{2}{3}, \frac{2}{3}, \frac{1}{3}, \frac{1}{3}, \frac{1}{3})$. The committee $\quickset{c_1, c_4, c_5}$, however, induces a load distribution of $(\frac{3}{5}, 0, \frac{3}{5}, \frac{3}{5}, \frac{3}{5}, \frac{3}{5})$ which is lower both in terms of the maximum as well as lexicographic ordering.
\end{example}

\begin{example}\label{ex:MinimaxAV}
Consider the following profile with $n=100$ voters and let $k=2$:
    \[ 99 \times \{ c_1,c_2 \} \qquad 1 \times \{ c_3, \ldots, c_8 \}. \]
Minimax-AV (which minimises the maximum Hamming-distance among all voters to the winning committee) selects any two candidates from $\{ c_3, \ldots c_8 \}$ and none that is supported by the 99 voters. This clearly violates individual representation and semi-strong JR.
\end{example}

\section{Common ABC Rules do not Guarantee $(2,4)$-IR  for VI Preferences} 
\label{app:inapr-rule-VI}

The two examples below show that many common ABC rules are not guaranteed to return a committee that provides $(2,4)$-IR for VI preferences.
\begin{example}
Consider the following profile with $n=28$ voters and let $k=14$:
\begin{align*}
    A_1&=\quickset{c_1, \ldots, c_7}\\
    A_2=\dots = A_{14}&=\quickset{c_1, \ldots, c_{20}}\\
    A_{15}=\dots=A_{28} &= \quickset{c_{8}, \ldots, c_{20}}.
\end{align*}
Notice that $f_1=7$ and $f_i = 13$ for $i > 1$. 
Here, MES, AV, PAV, seq-PAV, rev-seq-PAV, seq-Phragm{\'e}n, SAV, and Greedy Monroe (depending on the tie breaking) all return committees $W$ where  $\quickset{c_8, \ldots, c_{20}} \subset W$ and thus only one of the candidates among $\quickset{c_1, \ldots, c_7}$ is included. Hence, they all fail to select a $(2,4)$-IR committee.
\end{example}

\begin{example}
Consider the following profile with $n=28$ voters and let $k=14$:
\begin{align*}
    A_1= \dots = A_{14} &= \quickset{c_1, \ldots, c_7} \\
    A_{14+i} &= \quickset{c_{5+i}}, \qquad \forall i \in \quickset{1,\dots,14}.
\end{align*}
Again, we have $f_1 = 7$. CC and seq-CC return committees $W$ where only one candidate among $\quickset{c_1, \ldots, c_{7}} $ is included. Thus, they fail to select a $(2,4)$-IR committee.
\end{example}

A simple adaption of \Cref{ex:MinimaxAV} shows that Minimax-AV fails to provide any good approximation for IR even on $t$-PART instances. From Table~4.1 of the recent survey by \citet{LaSk22a}, the only rules missing are leximin-Phragmén and Monroe-AV for which we experimentally found profiles with $n=30$ voters and $m=300$ candidates where these two rules also fail to select a $(2,4)$-IR committee. 
Thus, we can conclude that \Cref{alg:ind-fair-DUE} outperforms all common ABC voting rules in terms of approximating IR on VI profiles.

\begin{table}[tb]
\begin{center}
\begin{tabular}{ c c c c}
\toprule
& \multicolumn{2}{c}{Individual Representation}   & semi-strong-JR \\ 
\cmidrule{2-3}
  & Lower Bound & Upper Bound \\ 
\midrule
 PART & $(1,0)$ & $(1,0)$ & \checkmark \\  
 $\alpha$-TR & $(1,0)$ & $(1,0)$ & \checkmark \\
 VEI & $(2-2/k,0)$ & $(2,0)$  & \checkmark\\ 
 CEI & $(2-2/k,0)$ & $(2,0)$  & \checkmark\\ 
 DUE & $(2-2/k,0)$ & $(2,4)$ & $\times$ \\
 VI & $(2-2/k,0)$ & $(2,4)$ & $\times$\\
 CI & $(1,k-1)$ & $(1,k-1)$  & $\times$\\
 \bottomrule
\end{tabular}
\caption{Individual representation guarantees in structured profiles. The column Lower Bound indicates that,
while the column Upper Bound gives us values $(\alpha, \beta)$ such that an $(\alpha, \beta)$-IR committee always exists (and can be computed efficiently). The last column refers to the existence of semi-strong JR committees in every instance.}  
\label{tab:IF_guarantees_on_structured_profiles}
\end{center}
\end{table}

\section{Further Domain Restrictions} \label{app:fur_dom_res}

We consider all the restricted domains that are discussed by \citet{ElLa15a} and \citet{yang2019tree}. An overview of how these domain restrictions are related to each other (which is adapted from \citet{yang2019tree}) can be found in \Cref{fig:relationship-restricted-domains}. An overview of the results of this section can be found in \Cref{tab:IF_guarantees_on_structured_profiles}.

\subsection{$t$-PART}

\begin{definition}[$t$-partition ($t$-PART)] An approval profile $A$ satisfies \emph{$t$-partition ($t$-PART)} if there is a partition $(C_1,\dots,C_t)$ of $C$ such that for every voter $i\in N$ there exists $C_j$ such that $A_i=C_j$.
\end{definition}
\begin{theorem}
Under $t$-PART approval profiles an IR committee always exists.
\end{theorem}

\vspace{-0.5cm}

\begin{proof}
Note that for each voter $i\in N$ with $A_i=C_j$ for some $j\in [t]$, $f_i=\floor{N(C_j) \cdot k/n}$. Now consider the committee that contains $\floor{N(C_j) \cdot k/n}$ candidates from each $C_j$. Clearly, $W$ provides IR and also
\begin{align*}
  |W|=  \sum_{j \in T} \left\lfloor N(C_j) \cdot \frac{k}{n} \right\rfloor \le n\cdot \frac{k}{n}=k,
\end{align*}
which completes the proof.
\end{proof}

 \begin{figure}[t]
      \centering
      \begin{tikzpicture}
      [dom/.style={minimum width = 0.7cm, rectangle, draw,  align=center },
      connec/.style={dotted,line width=0.3mm}, lab/.style={minimum width = 1cm,   align=center }]

        \node[dom] (PART) {$t$-PART };
        
        \node[dom, below left =0.5cm and  1.4cm of PART,label={\textbf{\textcolor{red}{$(1,k-1)$}}}, fill=gray!40] (aVTR) {$\tilde{\alpha}$-VTR};
        
        \node[dom, below = 0.72cm of aVTR,  fill=gray!40] (bVPTR) {$\tilde{\beta}$-VPTR};
        
         \node[dom, below = 0.47cm of bVPTR, fill=gray!40] (VPTR) {VPTR};
         
         \node[dom, below = 0.47cm of VPTR, fill=gray!40] (VTR) {VTR};
         
         \draw[line, connec] (aVTR.south)--(bVPTR.north);
         
         \draw[line, connec] (bVPTR.south)--(VPTR.north);
         \draw[line, connec] (VPTR.south)--(VTR.north);

        \node[dom, below right =0.5cm and 2.3cm of PART, label={\textbf{\textcolor{green!70!black}{$(1,0)$}}}] (aTR) {$\tilde{\alpha}$-TR };

        \node[dom, below = 0.72cm of aTR,  fill=gray!40] (bPTR) {$\tilde{\beta}$-PTR};
        \draw[line, connec] (aTR.south)--(bPTR.north);
        \node[dom, below = 0.47cm of bPTR,  fill=gray!40] (PTR) {PTR};
        \draw[line, connec] (bPTR.south)--(PTR.north);
         \node[dom, below = 0.47cm of PTR,  fill=gray!40] (TR) {TR};
         \draw[line, connec] (PTR.south)--(TR.north);

         \node[dom, below left = 0.9cm and -0.68cm  of PART ] (WSC) {WSC};
         \node[dom, below right = 0.9cm and -0.68cm  of PART] (VEI) {VEI};
         
        \node[dom, below right = 0.9cm and 0.6cm  of PART ] (CEI) {CEI };
        
         \node[lab, below right = 0.15 and -0.9cm  of PART] (help) {\textbf{\textcolor{red}{$(2-2/k,0)$}}};

             \draw [decorate, decoration = {brace,mirror}, color=green!70!black] (0.7,-1.65) --  (1.7,-1.65);
         
           \node[lab, below right = 0.650 and -1.4cm  of help,] (help2) {\textbf{\textcolor{green!70!black}{$(2,0)$}}};

       \draw [decorate, decoration = {brace}, color=red] (-1,-1) --  (2,-1);

         \node[dom, below  = 2.5cm  of PART,  fill=gray!40 ] (DUE) {DUE };
          \draw[line, connec] (WSC.south)--(DUE.north west);
        \draw[line, connec] (VEI.south)--(DUE.north east);
        \draw[line, connec] (CEI.south)--(DUE.east);

         \node[dom, below left = 1.2cm and 0.2cm  of DUE, label= east:\textbf{\textcolor{green!70!black}{$(2,4)$}},  fill=gray!40 ] (VI) {VI};
        \node[dom, below right = 1.2cm and 0.2cm  of DUE,  fill=gray!40 ] (CI) {CI};
       
        \draw[line, connec] (CI.north east)--(bPTR.south west);
        
        \draw[line, connec] (aVTR.south)--(CI.north west);
         \draw[line, connec] (aTR.south west)--(VI.north east);
         
         \draw[line, connec] (PART.south)--(DUE.north);
         
         \draw[line, connec] (DUE.south)--(VI.north);
         \draw[line, connec] (DUE.south)--(CI.north);

        \draw[line, connec] (PART.west)--(aVTR.east);
       \draw[line, connec]   (PART.east)--(aTR.west);
       \draw[line, connec] (VI.north west)--(bVPTR.south east);
      
      \end{tikzpicture}
     \caption{%
     Upper (green) and lower  (red) bounds of approximate IR and guaranteed existence (white) or not (gray) of  semi-strong JR committees under the respective domain. An arrow from X to Y shows that X implies Y. An upper  (resp. lower) bound for each domain is the best (resp. worst) upper (resp. lower) bound of any of its descendants (resp. ancestors), including itself.  
     }
    \label{fig:relationship-restricted-domains}
 \end{figure}

\subsection{CEI, VEI and WSC}

\begin{definition}[Candidate Extremal Interval (CEI)] An approval profile $A$ satisfies \emph{Candidate Extremal Interval (CEI)} if there is a linear order of $C$ such that for every voter  $i \in N$, the approval set $A_i$ forms a prefix or a suffix of that order. 
\end{definition}

\begin{definition}[Voter Extremal Interval (VEI)] An approval profile $A$ satisfies \emph{Voter Extremal Interval (VEI)} if there is a linear order of $N$ such that for every candidate $c_j \in C$, the set $N(\{c_j\})$ of voters approving $c_j$ forms a prefix or a suffix of that order. 
\end{definition}  

\begin{definition}[Weakly Single-Crossing (WSC)] An approval profile $A$ satisfies \emph{Weakly Single-Crossing (WSC)} if there is a linear order of $N$ such that
for each pair of candidates $c$,$c'$ in $C$ it holds that each
of the voter sets $N_1=\{i: c \in A_i, c' \not \in A_i \}$,    $N_1=\{i: c \not\in A_i, c'  \in A_i \}$ and $N_3= N\setminus (N_1 \cup N_2)$ forms an interval of this ordering, and $N_3$
appears between $N_1$ and $N_2$. 
\end{definition}

\begin{restatable}{theorem}{ubcei}
For any CEI and VEI profile, there exists a $(2,0)$-IR  and a semi-strong JR committee,  and both can be found in polynomial time.
\end{restatable}

\vspace{-0.5cm}

\begin{proof}
Assume we are given a CEI profile and the order $c_1, \dots, c_m$ over the candidates such that  each $A_i$ is either a prefix or a suffix of that order. If $f_i=k$ for some voter, then this means that there are at least $k$ candidates that are approved by all the voters and these candidates can form the winning committee which clearly is IR and semi-strong JR.  Otherwise, if $f_i \leq k -1$ for all voters $i \in [n]$,  we set 
\begin{align*}
\s = \{c_1, \dots, c_{\floor{k/2}}, c_{m -\ceil{k/2}}, \dots, c_m \},
\end{align*}
i.e., $\s$ consists of the first $\floor{k/2}$ candidates and the last $\ceil{k/2}$ candidates of the order. Hence, for each $i \in [n]$, it holds that $|A_i \cap \s|\geq \min\{k/2-1,|A_i|\}$, and as $f_i\leq k-1$, we get that $\s$ provides $(2,0)$-IR.  Now notice that as $\s$ consists of the first and the last candidate in the order. Hence,  $\s$ provides semi-strong-JR as  each voter approves one of these two candidates. 

Now, assume we are given a VEI profile and the order over the voters. %
We order the candidates as follows. Let $C_{p}$ and $C_{s}$ be the set of candidates such that the voters that approve a candidate in $C_p$ and  $C_{s}$ form a prefix and a suffix of $N$, respectively. Without loss of generality, a candidate $c$  such that $N(c)=N$ is assigned to $C_p$.  A candidate $c$ in $C_p$ is ordered before $c'$ in $C_{p}$ if the last voter that approves $c$ is ordered after the last voter that approves $c'$ (break ties arbitrarily). A candidate $c$ in $C_s$ is ordered before $c'$ in $C_{s}$ if the first voter that approves $c$ is ordered after the first voter that approves $c'$ (break ties arbitrary). All the candidates in $C_{p}$ are ordered before the candidates in $C_s$. As above, if $f_i=k$ for some voter, then this means that there are at least $k$ candidates that are approved by all the voters and these candidates form the winning committee which clearly is IR and semi-strong JR.   Otherwise, if  $\s$ consists of the first $\floor{k/2}$ candidates and the last $\ceil{k/2}$ candidates of the order that we describe above, then for each $i \in [n]$, it holds that $|A_i \cap \s|\geq \min\{k/2-1,|A_i|\}$, and as $f_i\leq k-1$, we again get that $\s$ provides $(2,0)$-IR. Now similarly as above, as $\s$ consists of the first and the last candidate of the order, $\s$ is  semi-strong-JR as if $f_i\geq 1$ for some $i\in N$, then $i$  approves some candidate between these two candidates. 
\end{proof}

\begin{restatable}{theorem}{ssJRinWSC}
For any WSC profile, there exists a  semi-strong JR committee and it can be found in polynomial time.
\end{restatable}

\vspace{-0.5cm}

\begin{proof}
Consider the order $1,\dots,n$ of the voters. Without loss of generality assume that for each $i \in N$, $|A_i|>1$ (otherwise we can just exclude them) and that $k \geq 2$. Now, let $c^*$ be a candidate that  is approved by voter $1$ and let $i>1$ be the first voter in the order that  does not approve $c^*$. This means that $i$ approves a candidate  $c'$ different from $c^*$, and the profile is WSC if all the voters from $i$ to $n$ approve $c'$. Hence, any committee $\s$ with $\{c^*,c'\} \subseteq \s$ provides semi-strong JR, as all the voters  are represented by at least one candidate.  
\end{proof}

\begin{restatable}{proposition}{lbCEI}
There exists an instance that is VEI, CEI and WSC, and does not admit a $(2-2/k,0)$-IR committee.
\label{lemm:lower-bound-CEI}
\end{restatable}

\vspace{-0.5cm}

\begin{proof}
Consider the instance defined in the proof of \Cref{lemm:VI-lower-bound}.   
We now show that the profile is CEI, VEI and WSC, from which the lemma follows. Indeed, for CEI, if we order the candidates as $c_1$, $c_2, \dots, c_m$, then each $A_i$  forms a prefix or a  suffix of the ordering. For VEI, if we order the voters as $1$, $2, \dots, n$, then the voters that approve each candidate  form a prefix or a  suffix of the ordering. Under the same ordering of the voters, note that for each candidate $c \in \cC$ we have $c \in A_i$ for either all $i \in [n-1]$ or for all $i \in [n] \setminus \quickset{1}$. Thus, the profile is also WSC. 
\end{proof}

\subsection{DUE}
\begin{definition}[Dichotomous Uniformly Euclidean (DUE)] An approval profile $A$ satisfies \emph{Dichotomous uniformly Euclidean (DUE)} if there is a mapping of voters and candidates into the real line and a radius $r$ such that  every voter $i \in \N$ approves the candidates that are at most $r$ far from her.
\end{definition}

\begin{restatable}{proposition}{lemmaDUEnossJR}
There exists a DUE profile which does not admit a semi-strong JR committee.
\end{restatable}

\vspace{-0.5cm}

\begin{proof}
Consider the following instance with $n=6$ voters and let $k=3$:
\begin{align*}
    A_1 &= \quickset{c_1}  &&A_4 = \quickset{c_3}\\
    A_2 &= \quickset{c_1,c_2}  &&A_5 = \quickset{c_3,c_4}\\
    A_3 &= \quickset{c_2}  &&A_6 = \quickset{c_4}.
\end{align*}
To see that this is a DUE profile, consider the following mapping of the instance onto the real line. Each voter $i \in [6]$ is mapped to the point $i$ and candidates $c_1$ to $c_4$ are mapped to  1.5, 2.5, 4.5 and 5.5, respectively. From this mapping we obtain the above profile by using an approval radius of $0.5$.
We observe $f_i=1$ for all $i \in [6]$, but for any committee $\s$ of size 3 it holds that $|A_i\cap \s|=0$ for some $i \in \{1, 3,4,6\}$.
\end{proof}

\subsection{$\tilde{\alpha}$-Tree Representation}

\begin{definition}[$\tilde{\alpha}$-tree representation ($\tilde{\alpha}$-TR)]
An approval profile $A$ satisfies \emph{$\tilde{\alpha}$-tree representation} ($\tilde{\alpha}$-TR)\footnote{We use $\tilde{\alpha}$ instead of $\alpha$ in order to avoid confusion with the use of $\alpha$ for denoting an approximation of IR. Same with $\tilde{\alpha}$-VTR, $\tilde{\beta}$-VPTR and $\tilde{\beta}$-PTR.} if there exists a rooted tree $T$ with vertices $C \cup \{x\}$ and root $x \notin C$ such that for every voter $i \in N$ there is a candidate $c \in C$ such that $A_i$ equals the set of vertices on the path from $x$ to $c$ (excluding $x$ but including $c$).
\end{definition}

Under $\tilde{\alpha}$-TR preferences, the following algorithm always selects a committee that provides individual representation. Here we use $\dist(c,x)$ to denote the (edge-)distance between the root $x$ and a candidate $c \in C$.
\begin{algorithm}[h]
	\caption{Individual representation for $\tilde{\alpha}$-TR}\label{alg:ind-fair-aTR}
	\begin{algorithmic}[1]

	    \State $W \gets \emptyset$
        \State	traverse the tree in any order
		\For {$c \in T$}
			\If{$|N(\{c\})| \geq \frac{n}{k} \cdot \dist(c,x)$}
				\State $W \gets W \cup \{c\}$	
			\EndIf
		\EndFor
		\If{$|W| < k$}\label{ln:aTR_notk}
				\State fill up $W$ arbitrarily	\label{ln:aTR_fill}
			\EndIf
		\\ \Return $W$
	\end{algorithmic}
\end{algorithm}

Note that this algorithm basically resembles HareAV as defined by \citet{ABC+16a} with a specific tie-breaking mechanism that depends on the tree representation. Due to this tie-breaking, it is clearly polynomial time computable.

\begin{restatable}{theorem}{alphatree}\label{thm:alphatree}
Under $\tilde{\alpha}$-TR preferences, a committee providing individual representation always exists and can be found in polynomial time.
\end{restatable}

\vspace{-0.5cm}

\begin{proof}
Let $T[W]$ be the subgraph of $T$ induced by the candidate set $W$ chosen by \Cref{alg:ind-fair-aTR}. 
\new{First, note that $T[W]$ is a subtree of $T$ with root $x$. To see this, let $c \in W$ and let $c'$ be the direct ancestor of $c$ (i.e., $c'$ is the candidate immediately before $c$ on the path from $x$ to $c$). By the $\tilde{\alpha}$-TR property we have $N(\{c\}) \subseteq N(\{c'\})$ and by the definition of the distance function it holds that $\dist(c,x) = \dist(c',x) + 1$. Thus $|N(\{c'\})| \geq \frac{n}{k} \cdot \dist(c',x)$ and $c'$ is also in $W$.
We now show that the committee $W$ is of size $k$. Then, we show that it indeed provides individual representation.}

By \Cref{ln:aTR_notk,ln:aTR_fill} of \Cref{alg:ind-fair-aTR} it holds that $|W| \geq k$. Now consider $T[W]$; we will assign $\frac{n}{k}$ voters to each edge of this tree. \new{In order to do this, for a leave $c$ of the tree choose $\frac{n}{k}$ (yet unassigned) voters from $N(\{c\})$ and assign them to the edge incident to $c$. (If $\frac{n}{k}$ is not an integer we assign one voter only partially.) Then we delete $c$ and that edge from the tree, starting this process again with a leave of the smaller tree.
Since $|N(\{c\})| \geq \frac{n}{k} \cdot \dist(c,x)$ and $N(\{c\}) \subseteq N(\{c'\})$ for all candidates $c'$ on the path from $x$ to $c$, there always exist $\frac{n}{k}$ such voters. Since there are only $n$ voters in total we have $|W| \leq k$.}

To show that satisfies IR, consider a group of voters $V \subseteq N$ such that $|V| \geq \ell \cdot \frac{n}{k}$ and $|\bigcap_{i \in V} A_i| \geq \ell$ for some $\ell \in \mathbb{N}$. By the definition of $\tilde{\alpha}$-TR there is a path of (edge-)length $\ell$ from $x$ to some $c \in C$ such that the candidates on that path (including $c$) are a subset of $\bigcap_{i \in V} A_i$ of size $\ell$. We call the set of these candidates $A_V \subseteq \bigcap_{i \in V} A_i$.
Every candidate in $A_V$ is approved by all of $V$, i.e., by at least $\ell \cdot \frac{n}{k}$ voters, and is at distance $\leq \ell$ from the root of $T$. Thus \Cref{alg:ind-fair-aTR} chooses all candidates in $A_V$ and therefore we have $|A_i \cap W| \geq \ell$ for all $i \in V$.
\end{proof}

\subsection{$\tilde{\alpha}$-VTR}

\begin{definition}[$\tilde{\alpha}$-vertex tree representation ($\tilde{\alpha}$-VTR)]
An approval profile $A$ satisfies \emph{$\tilde{\alpha}$-vertex tree representation} ($\tilde{\alpha}$-VTR) if there exists a rooted tree $T$ with vertices $V \cup \{x\}$ and root $x \notin V$, such that for every candidate $c_j \in C$ there exists a voter approving $c_j$ such that the set $N(\{c_j\}$ of voters approving $c_j$ forms a path from $x$ to that voter (excluding $x$ but including the voter).
\end{definition}
\begin{proposition}\label{prop:alpha-vtr}
There exists an $\tilde{\alpha}$-VTR instance  that does not admit an $(\alpha,\beta)$-IR committee for $\beta<k-1$ and any $\alpha \geq 1$.
\end{proposition}

\vspace{-0.5cm}

\begin{proof}
We note that the instance in \Cref{thm:no-beta-approx} is also a $\tilde{\alpha}$-VTR instance, from which the statement follows. To see this, consider a rooted tree with root $x$ and vertices $v \in V$ as follows: the voters $v_{\frac{n}{k}+1}, \ldots, v_n$ form a simple path starting from $x$. All the remaining voters $v_1, \ldots, v_{\frac{n}{k}}$ are incident to the vertex farthest away from $x$ and make up the leaves of the tree. The set of voters approving a common candidate in the above instance now form a path from $x$ to one of the leaves $v_1, \ldots, v_{\frac{n}{k}}$ and thus this is an $\tilde{\alpha}$-VTR representation of that instance. 
\end{proof}

Note that the instance discussed in the proof of \Cref{prop:alpha-vtr} does not admit a committee providing semi-strong JR.

\end{document}